%% file: main.tex
\documentclass[11pt]{article}
\input{pkgs}
\input{defs}

\input{environments}

\input{colors}
\usepackage[style=alphabetic,natbib=true,maxcitenames=2, maxbibnames=10,backend=bibtex]{biblatex}
\allowdisplaybreaks

\title{Regret Minimization for Piecewise Linear Rewards:\\Contracts, Auctions, and Beyond}

\author{
    Francesco Bacchiocchi\thanks{Politecnico di Milano, \texttt{\textcolor{blueGrotto}{francesco.bacchiocchi@polimi.it}}}\and
    Matteo Castiglioni\thanks{Politecnico di Milano, \texttt{\textcolor{blueGrotto}{matteo.castiglioni@polimi.it}}} \and
    Alberto Marchesi \thanks{Politecnico di Milano, \texttt{\textcolor{blueGrotto}{alberto.marchesi@polimi.it}}}
    \and
    Nicola Gatti \thanks{Politecnico di Milano, \texttt{\textcolor{blueGrotto}{nicola.gatti@polimi.it}}}
}

\date{}

\begin{document}

\maketitle
\input{abstract}

\pagenumbering{gobble} 

\clearpage
\tableofcontents
\renewcommand{\thmtformatoptarg}[1]{~#1}
\newpage
\pagenumbering{arabic}
\input{introduction}

\input{preliminaries}
\input{applications}
\input{no_regret}
\input{lower_bound}
\section*{Acknowledgments}

This paper is supported by the Italian MIUR PRIN 2022 Project ``Targeted Learning Dynamics:
Computing Efficient and Fair Equilibria through No-Regret Algorithms'', by the FAIR (Future
Artificial Intelligence Research) project, funded by the NextGenerationEU program within the PNRRPE-AI scheme (M4C2, Investment 1.3, Line on Artificial Intelligence), and by the EU Horizon project
ELIAS (European Lighthouse of AI for Sustainability, No. 101120237).

\newpage
\printbibliography

\newpage
\appendix

\input{appendix}

\end{document}

%% file: pkgs.tex
\usepackage{pgfplots}

\usepackage{fullpage}

\usepackage{hyperref}
\usepackage[utf8]{inputenc} 
\usepackage[T1]{fontenc}    
\usepackage{url}            
\usepackage{mathtools}
\usepackage{booktabs}       
\usepackage{amsfonts}       
\usepackage{nicefrac}       
\usepackage{microtype}      
\usepackage{xcolor}         

\usepackage{nicematrix}

\usepackage{tikz}
\usepackage{float}
\usetikzlibrary{shapes,backgrounds,patterns,calc,arrows,arrows.meta}
\usepackage{chemfig}

\catcode`\_=11
\definearrow5{-x>}{%
    \CF_arrowshiftnodes{#3}%
    \CF_expafter{\draw[}\CF_arrowcurrentstyle](\CF_arrowstartnode)--(\CF_arrowendnode)%
        coordinate[midway,shift=(\CF_ifempty{#4}{225}{#4}+\CF_arrowcurrentangle:\CF_ifempty{#5}{5pt}{#5})](line@start)%
        coordinate[midway,shift=(\CF_ifempty{#4}{45}{#4}+\CF_arrowcurrentangle:\CF_ifempty{#5}{5pt}{#5})](line@end)%
        coordinate[midway,shift=(\CF_ifempty{#4}{135}{#4}+\CF_arrowcurrentangle:\CF_ifempty{#5}{5pt}{#5})](line@start@i)%
        coordinate[midway,shift=(\CF_ifempty{#4}{315}{#4}+\CF_arrowcurrentangle:\CF_ifempty{#5}{5pt}{#5})](line@end@i);
    \draw(line@start)--(line@end);%
    \draw(line@start@i)--(line@end@i);%
    \CF_arrowdisplaylabel{#1}{0.5}+\CF_arrowstartnode{#2}{0.5}-\CF_arrowendnode
}
\catcode`\_=8

\usepackage{pict2e,picture,graphicx}
\usepackage[overload]{empheq}
\usepackage{microtype}
\usepackage{graphicx}
\usepackage{subcaption}
\usepackage{booktabs}
\usepackage{amsmath}
\usepackage{cases}
\usepackage{mathtools}
\usepackage{amsthm}
\usepackage{thm-restate}
\usepackage{enumitem}

\allowdisplaybreaks
\usepackage{xcolor}
\usepackage{nicefrac, xfrac}
\usepackage{xparse}
\usepackage[capitalize, noabbrev]{cleveref}
\usepackage{wrapfig}
\usepackage{bbm}
\usepackage{yfonts}
\usepackage{nicefrac} 
\usepackage{multirow}
\usepackage{bm}
\usepackage{bbm}
\usepackage{subcaption}

\usepackage[T1]{fontenc}

\hypersetup{
	colorlinks = true,
	linkcolor = typ_blue,
	citecolor = Sepia,
	linktocpage = true,
	urlcolor = darkGreen
}

\usepackage{comment}

\usepackage{colortbl}
\usepackage{cellspace}
\usepackage{mleftright}
\usepackage{algorithm} 
\usepackage{dsfont}
\usepackage{algpseudocode}
\usepackage{amssymb}

%% file: defs.tex
\newcounter{programCounter}
\crefname{programCounter}{Program}{Programs}

\makeatletter
\DeclareRobustCommand{\Arrow}[1][]{%
	\check@mathfonts
	\if\relax\detokenize{#1}\relax
	\settowidth{\dimen@}{$\m@th\rightarrow$}%
	\else
	\setlength{\dimen@}{#1}%
	\fi
	\sbox\z@{\usefont{U}{lasy}{m}{n}\symbol{41}}%
	\begin{picture}(\dimen@,\ht\z@)
		\roundcap
		\put(\dimexpr\dimen@-.7\wd\z@,0){\usebox\z@}
		\put(0,\fontdimen22\textfont2){\line(1,0){\dimen@}}
	\end{picture}%
}
\makeatother

\newcommand{\A}{\mathcal{A}}

\makeatletter
\newcommand{\pushright}[1]{\ifmeasuring@#1\else\omit\hfill$\displaystyle#1$\fi\ignorespaces}
\newcommand{\pushleft}[1]{\ifmeasuring@#1\else\omit$\displaystyle#1$\hfill\fi\ignorespaces}
\newcommand{\specialcell}[1]{\ifmeasuring@#1\else\omit$\displaystyle#1$\ignorespaces\fi}
\makeatother

\usepackage{pifont}
\usepackage[normalem]{ulem} %
\definecolor{mygreen}{rgb}{0.0, 0.5, 0.0}
\definecolor{myorange}{rgb}{1, 0.7, 0.1}

\newtheorem*{theorem-non}{Theorem}

%% file: environments.tex
\theoremstyle{plain}
\newtheorem{theorem}{Theorem}[section]

\newtheorem{corollary}[theorem]{Corollary}
\theoremstyle{definition}
\newtheorem{definition}[theorem]{Definition}

\theoremstyle{remark}

%% file: colors.tex
\definecolor{niceRed}{RGB}{190,38,38}
\definecolor{Red2}{RGB}{219, 50, 54}
\definecolor{mgreen}{RGB}{160, 200, 140}
\definecolor{blueGrotto}{RGB}{5,157,192}
\definecolor{limeGreen}{HTML}{81B622}
\definecolor{myellow}{rgb}{0.88,0.61,0.14}
\definecolor{darkGreen}{HTML}{2E8B57}
\definecolor{navyBlueP}{HTML}{03468F}
\definecolor{Sepia}{HTML}{7F462C}
\definecolor{red2}{HTML}{1F462C}
\definecolor{orange2}{HTML}{FF8000}
\definecolor{mgray}{HTML}{ABB3B8}
\definecolor{lgray}{HTML}{E5E8E9}
\definecolor{myPurple}{RGB}{175,0,124}
\definecolor{mypurple2}{rgb}{0.8,0.62,1}
\definecolor{royalBlue}{HTML}{057DCD}
\definecolor{mpink}{HTML}{FC6C85}
\definecolor{lblue}{RGB}{74,144,226}
\definecolor{peagreen}{RGB}{152,193,39}
\definecolor{typ_navy}{HTML}{001f3f}
\definecolor{typ_blue}{HTML}{0074d9}
\definecolor{typ_aqua}{HTML}{7fdbff}
\definecolor{typ_teal}{HTML}{39cccc}
\definecolor{typ_eastern}{HTML}{239dad}
\definecolor{typ_purple}{HTML}{b10dc9}
\definecolor{typ_fuchsia}{HTML}{f012be}
\definecolor{typ_maroon}{HTML}{85144b}
\definecolor{typ_red}{HTML}{ff4136}
\definecolor{typ_orange}{HTML}{ff851b}
\definecolor{typ_yellow}{HTML}{ffdc00}
\definecolor{typ_olive}{HTML}{3d9970}
\definecolor{typ_green}{HTML}{2ecc40}
\definecolor{typ_lime}{HTML}{01ff70}
\definecolor{newgreen}{HTML}{83c702}
\definecolor{newpurp}{RGB}{97,96,121}

%% file: abstract.tex
\begin{abstract}
Most microeconomic models of interest involve optimizing a \emph{piecewise linear function}.
	These include \emph{contract design} in hidden-action principal-agent problems, selling an item in \emph{posted-price} auctions, and bidding in \emph{first-price} auctions.
	When the relevant model parameters are \emph{unknown} and determined by some (unknown) probability distributions, the problem becomes \emph{learning how to optimize an unknown and stochastic piecewise linear reward function}.
	Such a problem is usually framed within an \emph{online learning} framework, where the decision-maker (learner) seeks to \emph{minimize the regret} of \emph{not} knowing an optimal decision in hindsight.
	This paper introduces a general online learning framework that offers a unified approach to tackle \emph{regret minimization for piecewise linear rewards}, under a suitable \emph{monotonicity} assumption commonly satisfied by microeconomic models.
	We design a learning algorithm that attains a regret of $\widetilde{O}(\sqrt{nT})$, where $n$ is the number of ``pieces'' of the reward function and $T$ is the number of rounds.
	This result is tight when $n$ is \emph{small} relative to $T$, specifically when $n \leq T^{\nicefrac{1}{3}}$.
	Our algorithm solves two open problems in the literature on learning in microeconomic settings.
	First, it shows that the $\widetilde{O}(T^{\nicefrac{2}{3}})$ regret bound obtained by~\citet{zhu2023sample} for learning optimal \emph{linear} contracts in hidden-action principal-agent problems is \emph{not} tight when the number of agent's actions is small relative to $T$.
	Second, our algorithm demonstrates that, in the problem of learning to set prices in posted-price auctions, it is possible to attain suitable (and desirable) instance-independent regret bounds, addressing an open problem posed by~\citet{cesa2019dynamic}.
\end{abstract}

%% file: introduction.tex
\section{Introduction}

Several microeconomic problems involve the optimization of a \emph{piecewise linear} function.
These include \emph{contract design} in hidden-action principal-agent problems (see, \emph{e.g.},~\citep{dutting2019simple}), selling items in \emph{posted-price} auctions~(see, \emph{e.g.},~\citep{kleinberg2003value}), and bidding in \emph{first-price} auctions under certain modeling assumptions~(see, \emph{e.g.},~\citep{daskalakis2022learning}).
For instance, in a posted-price auction, a seller proposes a price $p$ to a potential buyer, who has some (unknown) private valuation $v$ for the item being sold.
The goal of the seller is to maximize revenue.
This is equal to $p$ whenever $p \leq v$, since the buyer buys the item, while it is zero if $p > v$, since the item goes unsold in such a case.
Thus, the seller's revenue is clearly piecewise linear as a function of the price $p$.
Another example is a basic contract design problem, in which a principal is willing to incentivize an agent to exert effort on a project, by promising them a fraction $\rho$ of the project's earnings (or rewards).
Based on $\rho$, the agent chooses a hidden action $i$, which corresponds to a certain level of effort put into the project.
Therefore, the principal's utility to be maximized is equal to $(1- \rho) R_i$, where $R_i$ is the expected reward under the action $i$ chosen by the agent, given $\rho$.
This is clearly a piecewise linear function of $\rho$, with the ``pieces'' determined by the agent's actions.

Algorithmic game theory has been historically concerned with the design of efficient computational methods to solve optimization problems arising from microeconomic models, under the assumption that all their relevant parameters are \emph{known}.
In simple settings, this usually reduces to maximizing or minimizing some (known) piecewise linear function, a task that can be easily tackled by means of standard linear programming techniques.
In recent years, growing attention has been devoted to settings in which the model parameters are \emph{unknown} and determined by some (unknown) probability distributions.
Such settings beget the fundamental and challenging problem of \emph{learning how to optimize an unknown and possibly stochastic piecewise linear function}.
%

%
This problem is typically framed within an \emph{online learning} framework, where the learner (\emph{e.g.}, the principal in contract design or the seller in posted-price auctions) repeatedly faces the ``one-shot'' decision-making problem to gather feedback on relevant model parameters (see, \emph{e.g.}~\citep{cesa2019dynamic,zhu2023sample}).
%
%
At each round, the learner gets a reward for their decision, which is determined by the unknown and possibly stochastic piecewise linear reward function that characterizes the setting. 
%
%
Thus, the learner's objective becomes \emph{minimizing the regret} with respect to what they could have obtained with knowledge of the parameters defining such a function.
%

In this paper, we introduce a general online learning framework that offers a unified approach to tackle \emph{regret minimization for piecewise linear rewards}, under a suitable \emph{monotonicity} assumption commonly satisfied by microeconomic settings.
This framework enables us to address some open problems in the literature on learning in microeconomic models, as we describe below.

\subsection{Regret Minimization for Piecewise Linear Rewards}

We refer to our online learning framework as \emph{Bandit with Monotone Jumps} (BwMJ).
%
%
%
In an instance of BwMJ, the learner takes actions $\alpha \in [0,1]$.
The action space is partitioned into $n \in \mathbb{N}_+$ action intervals (unknown to the learner), denoted $\A_i$ for $i = 1, \ldots, n$.
Each interval corresponds to a distinct ``piece'' of the piecewise linear reward function that characterizes the instance.
Specifically, the learner's reward for any action $\alpha \in \A_i$ is a random variable, and the expectations of such variables define a linear function of $\alpha$ over the interval $\A_i$.
This is because the reward for an action $\alpha \in \A_i$ is the product between a (deterministic) term $\ell(\alpha)$, where $\ell$ is a given linear function shared across all intervals, and a sample from an interval-specific probability distribution $\nu_i$, whose expected value is $\mu_i$.
Therefore, the expected reward over the interval $\A_i$ is the function $\ell(\alpha) \mu_i$.
Note that any action defining the ``meeting point'' between two intervals $\A_i$ and $\A_{i+1}$ corresponds to a \emph{jump discontinuity}, where the learner's expected reward abruptly changes from $\ell(\alpha) \mu_i$ to $\ell(\alpha) \mu_{i+1}$.
These points are crucial elements characterizing a BwMJ problem.
Figure~\ref{fig:jumps} shows an example of BwMJ instance.
%

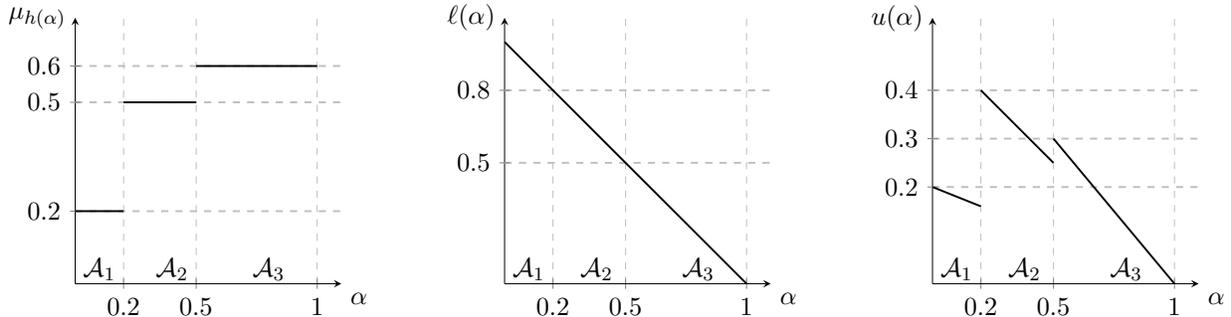
\begin{figure}[!t]
    \centering
    \resizebox{\linewidth}{!}{\input{tikz_plot/jumps.tikz}}
    \caption{Example of BwMJ instance. (\emph{Left}) The expected values $\mu_{h(\alpha)}$ as a function of $\alpha$, where $h(\alpha)$ is defined as the index $i$ such that $\alpha \in \A_i$. (\emph{Center}) The linear function $\ell(\alpha) \coloneqq 1-\alpha$. (\emph{Right}) The learner's expected reward $u(\alpha) \coloneqq \ell(\alpha) \mu_{h(\alpha)}$, which is a piecewise linear function over the action space $[0,1]$.}
    \label{fig:jumps}
\end{figure}

The learner interacts with the environment over $T$ rounds.
At each round $t$, they select an action $\alpha_t \in [0,1]$ and receive feedback in the form of a sample $x_t$ from the probability distribution $\nu_i$, where $\alpha_t \in \A_i$.
%
%
The learner only knows the linear function $\ell$, but they do \emph{not} know anything about the action intervals $\A_i$ and their corresponding probability distributions $\nu_i$, including their expectations $\mu_i$.
Thus, the learner's goal is to \emph{minimize the regret accumulated over $T$ rounds}, which is the difference between the expected reward attained by an \emph{optimal} action---computed with knowledge of the intervals $\A_i$ and their distributions $\nu_i$---and the expected reward obtained by the actions $\alpha_t$.
%

The BwMJ framework captures a variety of learning problems arising in microeconomic models.
For example, in posted-price auctions, it can model the problem of learning to set prices against unknown buyers, with the ``pieces'' of the piecewise linear reward function corresponding to the different possible (stochastic) buyer valuations~\citep{cesa2019dynamic}.
Moreover, in contract design, BwMJ instances capture the problem faced by a principal who must learn an optimal linear contract without knowing their own rewards and the agent's features, such as action costs and their outcome-determining probability distributions~\citep{zhu2023sample}.
%

Without imposing any conditions on the \emph{gaps} $\mu_{i+1}-\mu_i$ that characterize the jump discontinuities in a BwMJ instance, it is immediately clear that the regret-minimization problem is \emph{ill-posed}, as no algorithm can achieve sublinear regret (see Section~\ref{sec:learner_policies}).
This is why we introduce the \emph{monotonicity} property in the definition of BwMJ instances.
Specifically, we require that the learner's expected rewards are always increasing (or equivalently, decreasing) at the jumps.
In other words, we assume that all the gaps $\mu_{i+1}-\mu_i$ are positive or negative.
As we show in this paper, this condition is sufficient to enable the design of regret-minimization algorithms.
Furthermore, the monotonicity property does \emph{not} limit the modeling power of BwMJ instances, as it holds in most of the microeconomic settings captured by the framework.
For instance, in contract design, the jumps are characterized by positive gaps, whereas in posted-price auctions, the gaps are negative.
%

%

\subsection{Our Results and Techniques}

BwMJ problems can be seen as a special case of continuum-armed bandits with one-sided \emph{Lipschitz} rewards.
A natural regret-minimization algorithm for this setting is the one by~\citet{kleinberg2004nearly}, which guarantees an instance-independent regret bound of $\widetilde{O}(T^{\nicefrac{2}{3}})$. This algorithm uniformly discretizes the action space with step size $\epsilon \propto T^{\nicefrac{1}{3}}$ and runs a regret minimizer over the discretized set.
While this approach is optimal in the worst case, it becomes largely suboptimal in BwMJ instances where the number of jump discontinuities is small, as it may waste many rounds trying actions that are clearly suboptimal.

A more refined approach proposed by~\citet{cesa2019dynamic}---in the restricted setting of posted-price auctions---is to first attempt to identify the jump discontinuities and then run a regret minimizer only over the best action within each interval $\mathcal{A}_i$ (see Figure~\ref{fig:jumps}). However, given the stochastic nature of the feedback, detecting \emph{all} of the needed jump discontinuities is particularly challenging.
Indeed, the learner may incur large regret in order to identify jumps with excessively small gaps.  
%
%
Thus, the algorithm by~\citet{cesa2019dynamic} only guarantees an \emph{instance-dependent} regret bound that depends on the gap sizes and is largely suboptimal when the gaps are too small.
%
%
%

In this paper, we propose an approach for regret minimization in general BwMJ instances that avoids identifying jumps with small gaps.
%
%
Our approach consists of structuring the learning process into epochs.
%
%
In each epoch, our algorithm first searches for jumps with gaps larger than a suitable threshold. Once identified, the algorithm uses collected information about jumps to shrink the learner's action space by excluding actions that are clearly suboptimal. It then proceeds to the next epoch, halving the threshold used to identify jumps. The shrinking step is a crucial part of our procedure, as it ensures that in the subsequent epoch, the algorithm only considers actions that are not too suboptimal, thus effectively controlling the regret.
%

A key advantage of our method is that the resulting regret bound depends only on the \emph{number} of discontinuities, rather than their magnitudes. This leads to significantly improved performance in most BwMJ instances, \emph{i.e.}, when the number of jumps is small.
In particular, our learning algorithm guarantees a regret bound of $\widetilde{O} ( \sqrt{n T} )$, which we show to be tight when $n \leq T^{\nicefrac{1}{3}}$. Note that, as discussed above, when this condition is not satisfied, treating the problem as a generic one-sided Lipschitz bandit is the appropriate strategy.
%
%
More formally, we prove the following:
%
\begin{theorem}
    There exists an algorithm that attains $\widetilde{O} ( \sqrt{n T} )$ regret in any BwMJ instance, where $n$ is the number of action intervals of the instance and $T$ is the number of rounds of the learning interaction.
\end{theorem}

\begin{table}[!t]
    \centering
    \caption{Summary of the open problems solved in this paper. \emph{Notes:} ${\ddag}$ Regret bound for \emph{Bayesian} settings, where $d$ is the number of possible agent’s types. ${\dagger}$ The quantity $V$ is such that $V \propto \max_{i \in [n]} \nicefrac{1}{(v_{i+1} - v_{i})}$, where $v_i$ denotes the $i$-th buyer's valuation (see~\citep{cesa2019dynamic} for a formal definition).}
    \label{tab:results}
    \renewcommand{\arraystretch}{1.5}
    \begin{tabular}{c|c|c}
        \toprule
        Paper\,/\ Problem & Posted-price auctions & Principal-agent problems \\
         \hline
        \citet{zhu2023sample} & --- & $\widetilde{\mathcal{O}}(T^{\nicefrac{2}{3}})$ \\ \hline
        \citet{cesa2019dynamic} & $\widetilde{\mathcal{O}}(\sqrt{nT} + V(V+1))^\dagger$ & ---  \\  \hline
        \rowcolor{gray!20} Ours & $\widetilde{\mathcal{O}}(\sqrt{nT})$ & $\widetilde{\mathcal{O}}(\sqrt{nT}) \quad \widetilde{\mathcal{O}}(\sqrt{ndT})^{\ddag}$  \\ \bottomrule
    \end{tabular}
\end{table}
%
The above theorem addresses two major \emph{open questions} in the literature on learning in microeconomic models, as discussed below.
Table~\ref{tab:results} provides a summary of the results achieved.
%
%
\begin{itemize}
	\item \textbf{Learning optimal linear contracts in hidden-action principal-agent problems.} This learning problem was originally addressed by~\citet{zhu2023sample}, who provided an algorithm that achieves $\widetilde{O}(T^{\nicefrac{2}{3}})$ regret.
	\citet{zhu2023sample} also complement this result with a lower bound showing that such a regret guarantee is tight in $T$.
	However, their lower bound relies on problem instances in which the number $n$ of actions available to the agent grows proportionally to $T^{\nicefrac{1}{3}}$.
	Thus, they leave open the question of establishing whether better regret guarantees can be attained when the number of actions is small.
    %
	%
	%
	Our learning algorithm demonstrates that $\widetilde{O} ( \sqrt{n T} )$ regret can be attained in non-\emph{Bayesian} settings when the number of agent's actions $n$ is small relative to $T$, specifically when $n \leq T^{\nicefrac{1}{3}}$.
    We also provide a lower bound showing that such a regret guarantee is tight (Theorem~\ref{thm:lb}), complementing the lower bound by~\citet{zhu2023sample}.
    Furthermore, our learning algorithm also allows to achieve $\widetilde{O} ( \sqrt{n d T} )$ regret in \emph{Bayesian} settings, where $d$ is the number of possible agent's types.
    %
    %
	%
	\item \textbf{Learning in posted-price auctions with finitely many valuations.} This problem was originally addressed by~\citet{cesa2019dynamic}, who proposed an algorithm achieving a regret bound of $\widetilde{\mathcal{O}}(\sqrt{nT} + V(V+1))$ when cast within the BwMj framework, where $V$ is an instance-dependent parameter that encodes the difficulty of identifying jumps. 
    \citet{cesa2019dynamic} left as an open question establishing whether the dependence on $V$ in the regret bound is avoidable or \emph{not}.
   Since posted-price auctions can be framed as specific BwMJ instances, our learning algorithm shows that the dependence on $V$ in the regret bound achieved by~\citet{cesa2019dynamic} can indeed be avoided.
    %
\end{itemize}

\subsection{Related Works}
\paragraph{Online learning in contract design} 
Our paper is closely related to the line of research that considers repeated interactions in hidden-action principal-agent problems (see, \emph{e.g.},~\cite{dutting2024algorithmic} for a detailed survey). 
\citet{chien2016adaptive} first introduce the problem of learning approximately-optimal contracts in a setting where the principal has no prior knowledge of the agent and repeatedly interacts with them. 
They propose an algorithm with $\widetilde{\mathcal{O}}(\sqrt{m}\, T^{\nicefrac{2}{3}})$ regret with respect to an optimal linear contract, where $m$ is the number of outcomes. 
Subsequently, \citet{zhu2023sample} improve the result in \cite{chien2016adaptive} by providing an algorithm that achieves an $\widetilde{\mathcal{O}}(T^{\nicefrac{2}{3}})$ regret bound.
They also show that this bound is tight when the agent has a number of actions $n$ satisfying $n \propto T^{1/3}$. 
\citet{DuettingOptimal2023} show that it is possible to achieve $\mathcal{O}(\log\log T)$ regret when the principal receives stronger feedback than that considered by~\citet{zhu2023sample}.
Specifically, \citet{DuettingOptimal2023} assume that the principal observes their expected utility in each contract they propose to the agent.
In addition, \citet{zhu2023sample} also establish near-tight exponential regret bounds for general (bounded) contracts. 
Therefore, to overcome their negative result, \citet{bacchiocchi2024learning} and \citet{chen2024bounded} introduce additional assumptions and provide polynomial regret guarantees.

\paragraph{Dynamic Pricing and Posted Price Auctions}
Our work is also related to dynamic pricing and online posted price auctions. 
Over the last years, these problems have received significant attention from the computer science community (see, \emph{e.g.}, ~\cite{den2015dynamic}\ for a detailed survey). 
In contrast to the majority of previous works, in which the demand curve is a continuous function satisfying some regularity conditions (see, \emph{e.g.},~\cite{BesbesSurprising2015,Babaioff2015Dynamic}), our problem can be cast into a dynamic pricing problem with a piecewise constant demand curve.
In particular, \cite{SinglaWang2024} study sequential posted pricing under regular, \emph{i.e.}, continuous, valuation distributions, and their algorithm performs a binary-search procedure that shrinks a single interval across epochs. This is possible because regularity endows the seller’s expected utility with structural properties---half-concavity, a unique maximizer, monotonicity on each side of it, and Lipschitz continuity---that do not hold in our discrete framework.

\citet{denBoer2020Discontinuous} develop an algorithm for piecewise continuous demand curves, achieving a regret upper bound of $\mathcal{O}(C\sqrt{T} \log T)$ in the piecewise constant case.
Despite the fact that their problem is more general than the one captured by the BwMJ setting, the constant $C$ affecting their regret guarantees is an instance-dependent parameter.
Additionally, their approach requires prior knowledge of certain parameters, such as the number of discontinuities and the smallest drop in the demand curve.
\citet{cesa2019dynamic} improve the regret guarantees of~\citet{denBoer2020Discontinuous} when restricted to piecewise constant demand curves.
They show that it is possible to achieve a $\widetilde{\mathcal{O}}(\sqrt{nT} + V(V+1))$ regret bound, where $n$ is the number of buyer's valuations and the parameter $V$ is an instance-dependent parameter.
In the following, we show that it is possible to design an algorithm that provides instance-independent regret guarantees. 

\paragraph{Continuum-armed Bandits}
Our problem can be seen as a continuum-armed bandit problem~\cite{kleinberg2004nearly,Bubeck2008Online,Kleinberg2008bandits} in which the learner's expected utility is a one-sided Lipschitz function.
\citet{DuettingOptimal2023} consider a regret minimization problem when the learner's function is an unknown \emph{one-sided Lipschitz continuous function}.
However, unlike our paper, they receive deterministic feedback about the principal's utility after pulling an arm.
\citet{Jieming2018} study the problem of learning a Lipschitz function from binary feedback applied to contextual dynamic pricing. 
However, different from our paper, they consider an adversarial setting.
\citet{BalcanDickVitercik2018} study online learning for piecewise Lipschitz functions under a dispersion assumption, which requires discontinuities not to be overly concentrated across rounds. This assumption does not hold in our setting, since the same jump discontinuity may occur at every round and different jumps may be arbitrarily close. As a result, their dispersion-based guarantees may yield linear regret when applied to BwMJ instances.
\citet{pmlr-v206-mehta23a} introduce a threshold linear bandit problem in which, as in our problem, the learner’s utility is piecewise linear.
In their model, the learning action space coincides with $[0,1]^d$ (with $d \in \mathbb{N}$) and is divided into two regions by a separating hyperplane. 
In contrast, our work addresses scenarios with potentially many discontinuities, when the learner's action space is limited to the one-dimensional interval $[0,1]$.
Finally, \citet{LazzaroPikeBurke2025FixedBudget,LazzaroPikeBurke2025FixedConfidence} study change-point identification in piecewise constant bandits under fixed-budget and fixed-confidence objectives, respectively. Unlike our work, they focus on locating the discontinuities rather than minimizing cumulative regret.

%% file: tikz_plot/jumps.tikz
\begin{tikzpicture}
	\begin{axis}[
		name=plot1,
		axis lines=middle,
		xlabel={$\alpha$},
		xlabel style={below right},
		ylabel={$\mu_{h(\alpha)}$},
		ylabel style={left},
		xmin=0, xmax=1.1,
		ymin=0, ymax=0.733,
		xtick={0,0.2,0.5,1},
		ytick={0},
		grid=both,
		grid style={dashed,gray!50},
		extra y ticks={0.2, 0.5,0.6},
		extra y tick style={grid=major, grid style={thick, dashed}},
		width=5.5cm,
		height=5.5cm
		]
		\addplot[domain=0:0.2, samples=2, thick, color=black] {0.2};
		\addplot[domain=0.2:0.5, samples=2, thick, color=black] {0.5};
		\addplot[domain=0.5:1, samples=2, thick, color=black] {0.6};
		\node at (axis cs:0.10, 0.04) {$\mathcal{A}_1$};
		\node at (axis cs:0.4, 0.04) {$\mathcal{A}_2$};
		\node at (axis cs:0.8, 0.04) {$\mathcal{A}_3$};
	\end{axis}
	
	\begin{axis}[
		at={(plot1.east)},
		anchor=west,
		xshift=2.4cm, 
		axis lines=middle,
		xlabel={$\alpha$},
		xlabel style={below right},
		ylabel={$\ell(\alpha)$},
		ylabel style={left},
		xmin=0, xmax=1.1,
		ymin=0, ymax=1.1,
		xtick={0,0.2,0.5,1},
		ytick={0},
		grid=both,
		grid style={dashed,gray!50},
		extra y ticks={0.5, 0.8},
		extra y tick style={grid=major, grid style={thick, dashed}},
		width=5.5cm,
		height=5.5cm
		]
		\addplot[domain=0:1, samples=100, thick, color=black] {(1-x)};
		\node at (axis cs:0.10, 0.06) {$\mathcal{A}_1$};
		\node at (axis cs:0.38, 0.06) {$\mathcal{A}_2$};
		\node at (axis cs:0.8, 0.06) {$\mathcal{A}_3$};
	\end{axis}
	
	\begin{axis}[
		at={(plot1.east)},
		anchor=west,
		xshift=8.7cm, 
		axis lines=middle,
		xlabel={$\alpha$},
		xlabel style={below right},
		ylabel={$u(\alpha)$},
		ylabel style={left},
		xmin=0, xmax=1.1,
		ymin=0, ymax=0.55,
		xtick={0,0.2,0.5,1},
		ytick={0},
		grid=both,
		grid style={dashed,gray!50},
		extra y ticks={0.2, 0.3, 0.4},
		extra y tick style={grid=major, grid style={thick, dashed}},
		width=5.5cm,
		height=5.5cm
		]
		\addplot[domain=0:0.2, samples=2, thick, color=black] {0.2*(1-x)};
		\addplot[domain=0.2:0.5, samples=2, thick, color=black] {0.5*(1-x)};
		\addplot[domain=0.5:1, samples=2, thick, color=black] {0.6*(1-x)};
		\node at (axis cs:0.10, 0.03) {$\mathcal{A}_1$};
		\node at (axis cs:0.38, 0.03) {$\mathcal{A}_2$};
		\node at (axis cs:0.8, 0.03) {$\mathcal{A}_3$};
	\end{axis}
\end{tikzpicture}

%% file: preliminaries.tex
\section{Bandits with Monotone Jumps}\label{sec:learning_problem}

In this section, we formally introduce the \emph{Bandit with Monotone Jumps} (BwMJ) problems addressed in this paper.
As we discuss extensively in the following Section~\ref{sec:applications}, BwMJ problems provide a general framework that captures several learning tasks arising from microeconomic models, such as posted-price auctions and algorithmic contract design in principal-agent problems.

\subsection{Model}

%

%


An instance of BwMJ is defined as a triplet $\mathcal{I} \coloneqq \left((\mathcal{A}_i)_{i \in [n]},( \nu_i)_{i \in [n]}, \ell \right)$,\footnote{In this paper, we denote by $[n]$ the set of the first $n \in \mathbb{N}_+$ natural numbers, namely $[n] \coloneqq \{ 1, \ldots, n \}$.} where:
%
%
\begin{itemize}
	\item $(\mathcal{A}_i)_{i \in [n]}$ is a collection of action intervals such that $\mathcal{A}_i\coloneqq [\overline\alpha_i,\overline\alpha_{i+1} )$ for every $i \in [n-1]$ and $\mathcal{A}_n\coloneqq [\overline\alpha_n,\overline\alpha_{n+1} ]$, with $\overline\alpha_1 = 0$, $\overline\alpha_{n+1} = 1$, and $\overline\alpha_i < \overline\alpha_{i+1}$ for all $i \in [n]$.
	%
	%
	%
	\item $(\nu_i)_{i \in [n]}$ is a collection of probability distributions in which each distribution $\nu_i$ has support $\text{supp}(\nu_i) \subseteq [0,1]$ and expected value $\mu_i \in [0,1]$, with $\mu_i < \mu_{i+1}$ for all $i \in [n-1]$.\footnote{Notice that, for ease of presentation, we focus on BwMJ instances in which $\mu_i < \mu_{i+1}$ for all $i \in [n-1]$. However, our results can be easily generalized to instances where $\mu_i > \mu_{i+1}$ for all $i \in [n-1]$.}
	%
	%
	%
	%
	\item $\ell: [0,1] \to [0,1]$ is a strictly decreasing linear function.\footnote{Notice that, since the sequence of expected values $(\mu_i)_{i \in [n]}$ is strictly increasing, we can focus w.l.o.g.~on BwMJ instances in which $ \ell: [0,1] \to [0,1] $ is a strictly decreasing linear function. Otherwise, as it will be clear later, the optimal action for the learner would trivially be $\alpha=1$. Moreover, let us remark that our results can be extended to deal with BwMJ instances in which $\ell : [0,1] \to [0,1]$ belongs to a broader class of functions. However, since linear functions are sufficient to encompass the main applications presented in Section~\ref{sec:applications}, in this paper we only deal with linear functions, for ease of presentation.}
	%
\end{itemize}
In a BwMJ problem, the learner only knows the linear function $\ell$, while they do \emph{not} know anything about the intervals $\A_i$ and their corresponding probability distributions $\nu_i$, including their supports and expected values $\mu_i$.
Let us also remark that the terminology ``monotone jumps'' refers to the jump discontinuities characterizing the expected rewards of the learner, as described later.
%

The action space available to the learner is the interval $[0,1]$.
Whenever the learner selects an action $\alpha \in [0,1]$, they receive as feedback a sample from a probability distribution $\nu_i$, where $i \in [n]$ is the index of the interval containing $\alpha$, namely $\alpha \in \mathcal{A}_i$.
%
%
Thus, every time the learner selects an action belonging to some interval $\mathcal{A}_i$, they receive as feedback a sample from the same distribution $\nu_i$.
For ease of notation, we introduce a function $h: [0,1] \to [n]$ that takes $\alpha \in [0,1]$ as input and returns the index $i \in [n]$ such that $\alpha \in \mathcal{A}_i$.
Formally, for every $\alpha \in [0,1]$, it holds that:
\begin{equation*}
    h(\alpha) :=\sum_{i \in [n]} \mathbbm{I}\{ \alpha \in \mathcal{A}_i \} \cdot i.
\end{equation*}
Notice that the function $h$ is always well defined, as the intervals $(\mathcal{A}_i)_{i \in [n]}$ are disjoint by definition.

After selecting an action $\alpha \in [0,1]$ and observing feedback $x \sim \nu_{h(\alpha)}$, the learner achieves a reward (or utility) equal to the product $\ell(\alpha) \cdot x$.
We introduce a function $u : [0,1] \to \mathbb{R}$ to compactly denote the learner's expected reward as a function of their action $\alpha \in [0,1]$.
Formally:
\begin{equation*}
	u(\alpha) \coloneqq  \ell(\alpha) \,  \mathbb{E}_{x \sim \nu_{h(\alpha)}}[x] = \ell(\alpha) \, \mu_{h(\alpha)}.
\end{equation*}
Notice that, even if the function $ \ell: [0,1] \to [0,1] $ is linear, the learner's expected reward $u(\alpha)$ is in general a \emph{piecewise linear} function, and, thus, it may be \emph{discontinuous}.
Indeed, the function $u(\alpha)$ is characterized by $n-1$ jump discontinuity points, corresponding to actions defining the ``meeting point'' of two consecutive intervals $\mathcal{A}_i$ and $\mathcal{A}_{i+1}$, where the value of $u(\alpha)$ abruptly ``jumps'' due to the fact that $\mu_{h(\alpha)}$ changes from $\mu_i$ to $\mu_{i+1}$.
We call such actions \emph{(monotone) jumps} since, by assumption, the value of $\mu_{h(\alpha)}$ can only increase as $\alpha$ increases, given that $\mu_{i+1} - \mu_i > 0$.
Moreover, we refer to $\mu_{i+1} - \mu_i$ as the \emph{gap} of the jump at the ``meeting point'' of the intervals $\A_i$ and $\A_{i+1}$.

%
%
%


%

Figure~\ref{fig:jumps} graphically depicts an example of BwMJ instance.

\subsection{Learner's Policies and Regret}\label{sec:learner_policies}

In a BwMJ problem, the learner repeatedly interacts with the environment over $T \in \mathbb{N}_+$ rounds.

At each round $t \in [T]$, the learner-environment interaction goes as follows:
%
%
\begin{enumerate}[noitemsep,nolistsep]
	\item The learner selects an action $\alpha_t$ from the interval $[0,1]$.
	\item Given the action selected by the learner, the environment samples $ x_t \sim \nu_{h(\alpha_{t})}$.
	\item The learner observes $x_t$ as feedback and achieves reward $\ell(\alpha_t) \cdot x_t$.
\end{enumerate}



The learner's behavior over the $T$ rounds is described by a \emph{(deterministic) policy}, which is defined as a tuple $\pi \coloneqq (\pi_t)_{t \in [T]}$ whose components $\pi_t : \mathcal{H}_{t-1} \to  [0,1]$ encode how the learner selects actions at each round.
Specifically, $\pi_t$ is a mapping from the history of observations up to the preceding round $t-1$, namely ${H}_{t-1} \coloneqq (\alpha_0,x_0,\alpha_1 \ldots, \alpha_{t-1},x_{t-1})\in\mathcal{H}_{t-1}$, to an action $\alpha_{t}=\pi_t(H_{t-1})$, where the set $\mathcal{H}_{t-1} $ contains all the possible histories of observations up to round $t -1$.

The performance of a learner's policy $\pi\coloneqq (\pi_t)_{t \in [T]}$ over the $T$ rounds is evaluated in terms of the \textit{(cumulative) pseudo-regret} (henceforth sometimes called \emph{regret} for short), which is defined as:
\begin{equation*}
	R_T(\pi) := T \cdot \textnormal{OPT} - \mathbb{E}\left[\sum_{t \in [T]} u(\alpha_t) \right],
\end{equation*}
where the expectation is taken with respect to the randomness of the environment that generates the feedback received by the learner at each round.

In the regret definition, $\text{OPT}$ represents the learner's expected reward for an action that is \emph{optimal in hindsight}.
This is defined as $\text{OPT}\coloneqq \max_{\alpha \in [0,1]}u(\alpha)$.\footnote{Notice that the learner's expected reward \emph{always admits a maximum} in a BwMJ instance. This follows from the fact that the linear function $\ell$ is decreasing and the intervals $\mathcal{A}_i$ are closed on the left. Thus, within each interval $\mathcal{A}_i$, the maximum exists and is attained at the left extreme of the interval. Since the union of all the intervals $\mathcal{A}_i$ gives the whole space of actions $[0,1]$, the learner's expected reward clearly admits a maximum.}
In the following, for any $\epsilon \in [0,1]$, we say that an action $\alpha \in [0,1]$ is $\epsilon$-optimal if $u(\alpha) \ge \text{OPT} - \epsilon$.
Furthermore, in the rest of this paper, we sometimes omit the dependency on the policy from the regret, by simply writing $R_{T}(\pi)$ as $R_{T}$, whenever the policy $\pi$ is clear from context.

The goal of the learner is to find a policy that minimizes the pseudo-regret.
Ideally, the learner would like a policy achieving \emph{no-regret}, meaning that the regret grows sublinearly in the number of rounds $T$, namely $R_T = o (T)$.
This ensures that the per-round regret $\frac{R_T}{T}$ goes to zero as $T$ grows.

We conclude the section with two observations related to the attainability of the learner's goal.

\paragraph{Why the jump monotonicity condition is needed?}
As a first observation, let us remark that, without assuming that the expected values $\mu_i$ constitute a monotone sequence, namely $\mu_i < \mu_{i+1}$ for all $i \in [n-1]$ (or equivalently, $\mu_i > \mu_{i+1}$ for all $i \in [n-1]$), it is \emph{not} possible to attain no-regret.
%
%
This is because it is always possible to construct a BwMJ instance $\mathcal{I} \coloneqq \left((\mathcal{A}_i)_{i \in [n]},( \nu_i)_{i \in [n]}, \ell \right)$ in which the learner can obtain a utility greater than zero only over an arbitrarily small interval.
Indeed, it is sufficient to consider an instance in which there exists an index $i \in [n]$ such that $\mu_i = 1$ and $ \mathcal{A}_i = [\alpha^\star, \alpha^\star + \epsilon)$ for $\alpha^\star \in (0,1)$ and an arbitrarily small $\epsilon > 0$, while $\mu_j = 0$ for all the other indexes $j \ne i$.
Clearly, since the learner does \emph{not} know $\alpha^\star$, the probability with which any algorithm selects an action belonging to $\mathcal{A}_i$ can be made arbitrarily close zero as $\epsilon$ goes to zero.
This results in any algorithm suffering a pseudo-regret that grows linearly in the number of~rounds~$T$.
%
%
%

\paragraph{Why existing algorithms are not enough?}
As a second observation, let us remark that a na\"ive approach to achieve no-regret in BwMJ problems would be to exploit existing algorithms for continuum-armed bandit settings with one-sided \emph{Lipschitz} rewards.
Indeed, it is easy to check that, in BwMJ problems, the learner's expected reward function $u(\alpha)$ is always one-sided \emph{Lipschitz}.
%
%
As a result, it is possible to employ the algorithm by~\citet{kleinberg2004nearly} for continuum-armed bandits to achieve pseudo-regret of the order of $ \widetilde{\mathcal{O}}(T^{\nicefrac{2}{3}}) $.
This regret bound is optimal for the general bandit instances considered in~\citep{kleinberg2004nearly}.
However, as we show in the rest of this paper, in BwMJ problems it is possible to improve the regret bound when the number of discontinuities is sufficiently small, by employing an \emph{ad hoc} approach that substantially differs from the one adopted by~\citet{kleinberg2004nearly}, so as to attain a regret bound of the order of $\mathcal{\widetilde{O}}(\sqrt{nT})$.

%% file: applications.tex
\section{Bandits With Monotone Jumps Applied to Learning in Microeconomics Settings}\label{sec:applications}

In this section, we show that the BwMJ framework captures several learning tasks arising from microeconomic models.
In Section~\ref{sec:applications_contracts}, we draw a connection between BwMJ and \emph{hidden-action principal-agent}~problems.
Specifically, we show that attaining no-regret against an optimal-in-hindsight \emph{linear} contract in a hidden-action principal-agent problem---a challenge faced by~\citet{zhu2023sample}---is equivalent to attaining no-regret in a suitably-defined BwMJ instance.
In Section~\ref{sec:applications_postedprice}, we provide an analogous connection with learning problems arising in \emph{posted-price} auctions with finitely many buyer valuations~\citep{kleinberg2003value}.
%
%
%
%
%
%
In Section~\ref{sec:applications_firstprice}, we provide an additional application in the context of \emph{first-price} auctions.

\subsection{No-Regret Learning in Hidden-Action Principal-Agent Problems}\label{sec:applications_contracts}


In \emph{hidden-action principal-agent} problems, a principal seeks to incentivize an agent to take a costly unobservable action that leads to favorable outcomes, through the provision of payments.
%

An instance of hidden-action principal-agent problem is defined by a tuple $(A, \Omega, r, F, c)$, where $A \coloneqq [n]$ is a set of $n$ agent's actions and $\Omega \coloneqq [m]$ is a set of $m$ outcomes.
%
%
Each outcome $j \in [m]$ is associated with a reward $r_j \in [0,1]$ for the principal,\footnote{The rewards in the definition of hidden-action principal-agent problem should \emph{not} be confused with those in the definition of BwMJ problem. Indeed, the latter correspond to principal's utilities in principal-agent problems, as shown in this section.} while each action $i \in [n]$ is characterized by a cost $c_i \in [0,1]$ for the agent---with $c_1 = 0$---and a probability distribution $F_i $ over outcomes.
%
%
%
For every $i \in [n]$, we let $R_i \coloneqq \mathbb{E}_{j \sim F_i}[r_j]$ be the principal's expected reward when the agent selects $i$.

The principal incentivizes the agent to take a desirable action by committing to a \emph{contract}, which specifies a payment $p_j \ge 0$ from the principal to the agent for every possible outcome $j\in [m]$.
%
%
We focus on the special subclass of \emph{linear contracts}, which simply transfer a fraction $\rho \in [0,1]$ of principal's rewards to the agent, namely $p_j = \rho  r_j$ for every $j \in [m]$.
Thus, under a linear contract parametrized by $\rho \in [0,1]$, if the agent selects an action $i \in [n]$, then the principal's expected utility is equal to $ (1-\rho)R_i $, while the agent's one is $ \rho R_i - c_i$. 
After the principal commits to a linear contract $\rho \in [0,1]$, the agent chooses a \emph{best response}, which is an action $i^\star(\rho) \in [n]$ maximizing their expected utility under $\rho$.
When the agent is indifferent among multiple actions, we assume that ties are broken in favor of the principal, as it is standard in the literature (see, \emph{e.g.},~\citep{dutting2019simple}).
%
Thus, the principal's expected utility is $u(\rho) \coloneqq (1-\rho) R_{i^\star(\rho)} $ as a function of $\rho \in [0,1]$.

Figure~\ref{fig:contract_reward}~(\emph{Left}) shows an example of principal's expected utility $u(\rho)$.

In the following, for the sake of presentation, we assume w.l.o.g.~that agent's actions are ordered in increasing order of expected reward, so that $R_i < R_{i+1}$ for every $i \in [n-1]$.
%
%
%
Moreover, we assume that all the actions are implementable by at least one linear contract, which means that, for every action $i \in [n]$, there exists some $\rho \in [0,1]$ such that $\rho R_i - c_i > \rho R_j - c_j$ for all $j \in [n]$ with $j \ne i$.

When the principal-agent interaction is repeated for $T \in \mathbb{N}_+$ rounds, at each round $t \in [T]$:
%
%
%
\begin{enumerate}
	\item The principal commits to a linear contract $\rho_t \in [0,1]$. 
	\item The agent plays a best response $i^\star(\rho_t)$, which is \emph{not} observed by the principal.
	\item The principal observes the outcome $j \sim F_{i^\star(\rho_t)}$ realized by the agent’s action.
	\item The principal achieves a utility equal to $(1-\rho_t)r_j$, while the agent gets $\rho_t r_j - c_{i^\star(\rho_t)}$.
\end{enumerate}
In such a repeated setting, a learning algorithm for the principal has to prescribe a linear contract to commit to at each round, without any knowledge about costs and probability distributions associated with agent's actions.
The goal is to design algorithms whose regret $T \cdot u(\rho^\star) - \mathbb{E} \left[ \sum_{t \in [T]} u(\rho_t)\right]$ grows sublinearly in the number of rounds $T$, where $\rho^\star \in \arg\max_{\rho \in [0,1]} u(\rho)$ denotes an optimal-in-hindsight linear contract (see, \emph{e.g.},~\citep{zhu2023sample}).
This is equivalent to achieving no-regret in a suitable BwMJ instance, which we describe in the following.
%

First, let us observe that agent's best responses induce a partition of linear contracts.
Indeed, by letting $\mathcal{B}_{i} \coloneqq \left\{\rho \in [0,1] \mid i^\star(\rho)=i \right\}$ be the set of contracts under which action $i \in [n]$ is a best response, it is easy to see that $\mathcal{B}_{i} = [\beta_i, \beta_{i+1})$ for every $i \in [n]$, for a suitable sequence $(\beta_{i})_{i \in [n+1]}$ such that $\beta_i < \beta_{i+1}$ for all $i \in [n]$, $\beta_1=0$, and $\beta_{n+1}=1$ (see, \emph{e.g.},~\citep{dutting2019simple}).
%
%
%
%

Then, the BwMJ instance can be defined as follows.
\begin{itemize}
	\item The intervals $(\mathcal{A}_i)_{i \in [n]}$ coincide with the sets $\mathcal{B}_i = [\beta_{i}, \beta_{i+1})$.
	\item The distributions $(\nu_i)_{i \in [n]}$ are discrete and supported over all the possible principal's rewards $(r_j )_{j \in [m]}$, with the probability that $\nu_i$ assigns to value $r_j$ being equal to $F_{i,j}$. 
	%
	%
	%
	%
	\item The linear function $\ell: [0,1] \to [0,1]$ is such that $\ell(\rho)=1-\rho$ for every $\rho \in [0,1]$.
\end{itemize}
Finally, let us observe that the monotonicity of the jumps is guaranteed by the fact that, by definition, $\mu_i = \mathbb{E}_{j \sim F_i} [ r_j ] = R_i$, and the sequence of principal's expected rewards $(R_i)_{i \in [n]}$ is increasing.

%
%
%

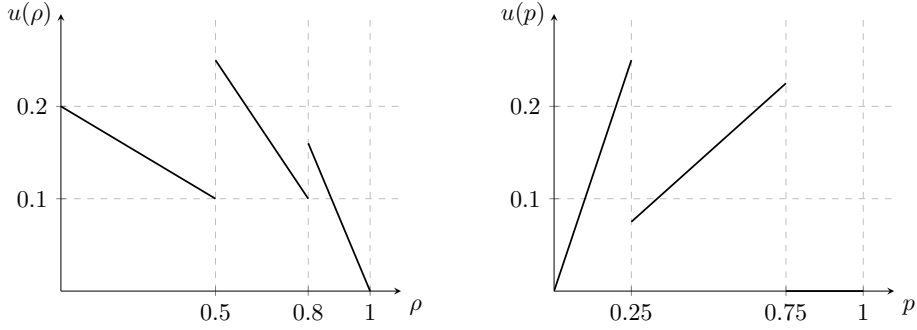
\begin{figure}[!t]
	\centering
	\resizebox{0.75\linewidth}{!}{\input{tikz_plot/contract_reward.tikz}}
	\caption{(\emph{Left}) Example of principal's expected utility by using linear contracts in a principal-agent problem. (\emph{Right}) Example of seller's expected utility in a posted-price auction.}
	\label{fig:contract_reward}
\end{figure}

\paragraph{Bayesian hidden-action principal-agent problems}
BwMJ problems also allow to encompass \emph{Bayesian} principal-agent settings in which the agent may be of different \emph{types}. 
Specifically, there is a set $\Theta \coloneqq [d]$ of $d$ agent's types, and both the costs $c_{i,k}$ and the probability distributions $F_{i,k}$ associated with agent's actions depend on the actual agent's type $k \in [d]$.
The goal is to attain no-regret against an optimal-in-hindsight linear contract when agent's types are drawn according to some fixed, but unknown, probability distribution $\mathcal{D}$. 
It is possible to show that an equivalent BwMJ instance can be built by partitioning $[0,1]$ into $nd$ intervals $\mathcal{B}_{i,k}$, each one corresponding to the set of contracts under which an agent of type $k \in [d]$ plays an action $i \in [n]$ as best response.

\subsection{No-Regret Learning in Posted-Price Auctions With Finitely Many Buyer Valuations}\label{sec:applications_postedprice}

In a \emph{posted-price} auction, a seller seeks to sell an item to a buyer, whose valuation for the item is sampled according to some unknown probability distribution (see, \emph{e.g.}, the random model in~\citep{kleinberg2003value}).
Formally, the seller posts a price $p \in [0,1]$ for the item, while the buyer's valuation $v \in (0,1)$ for the item is sampled from a probability distribution $\mathcal{V}$, which has support $\text{supp}(\mathcal{V}) \subseteq (0,1)$.
Then, the buyer chooses to acquire the item only if their valuation $v$ is greater than or equal to the price $p$ proposed by the seller.
Consequently, the seller's expected utility is equal to $u(p) \coloneqq p \cdot \mathbb{P}_{v \sim \mathcal{V}} \{ v \geq p \} = p \cdot (1- {F}(p)) $ for every $p \in [0,1]$, where $F: [0,1] \to [0,1]$ denotes a cumulative density function of $\mathcal{V}$ such that $F(x) \coloneqq \mathbb{P}_{v \sim \mathcal{V}} \{ v < x \}$ for all $x \in [0,1]$.

Figure~\ref{fig:contract_reward}~(\emph{Right}) depicts an example of seller's expected utility $u(p)$.

In the following, we assume that $\mathcal{V}$ is a discrete probability distribution such that $|\text{supp}(\mathcal{V})| = n$ for some $n \in \mathbb{N}_{+}$, 
and we define the support of $\mathcal{V}$ as $\text{supp}(\mathcal{V}) := \{v_i\}_{i=1}^n$.
%

When a posted-price auction is repeated for $T \in \mathbb{N}_+$ rounds, at each round $t \in [T]$: 
%
%
\begin{enumerate}
\item The seller proposes a price $ p_t \in [0,1] $ to the buyer.
\item The buyer’s valuation $ v_t \sim \mathcal{V} $ is sampled according to the distribution $ \mathcal{V} $.
\item If $ v_t \geq p_t $, then the buyer buys the item, the seller gains $ p_t $, and the buyer gets utility $ v_t - p_t $. Otherwise, if $ v_t < p_t $, both the seller and the buyer get utility zero. 
\end{enumerate}
Notice that, at the end of each round, the seller does \emph{not} observe $v_t$, but only whether the item has been sold or not.
In such a setting, the seller has \emph{no} knowledge about the distribution $\mathcal{V}$.
Thus, the learning task that they face is to minimize the regret with respect to the choice of always proposing an optimal-in-hindsight price, which is computed with knowledge of $\mathcal{V}$.

The learning task introduced above can be equivalently formulated in the framework of BwMJ problems.
%
%
Specifically, it can be captured by means of a BwMJ instance defined as follows.
\begin{itemize}
	\item There are $n+1$ intervals $(\mathcal{A}_i)_{i \in [n+1]}$ such that $\A_1 \coloneqq [v_0,v_1]$, with $v_0 \coloneqq 0$, and $\mathcal{A}_i \coloneqq (v_{i-1}, v_{i}]$ for all $i \in [n+1]$, with $v_{n+1} \coloneqq 1$.
	Intuitively, for all the prices in an interval $\mathcal{A}_i$, the item is sold with the same probability.
	Specifically, this is because, for all the valuations $v_j \in \text{supp}(\mathcal{V})$ with $j \geq i$, the buyer buys the item, while, for all the other valuations, they do \emph{not} buy. 
	\item The probability distributions $(\nu_i)_{i \in [n+1]}$ are \emph{Bernoulli} with parameters (\emph{i.e.}, expected values) $\mu_i \coloneqq 1 - F(v_{i})$.
	%
	Let us remark that $\mu_i$ is the probability of selling the item when the posted price lies within the interval $\mathcal{A}_i$, as $F(v_{i}) = \mathbb{P}_{v \sim \mathcal{V}} \{ v < v_{i} \}$ and $1 - F(v_{i}) = \mathbb{P}_{v \sim \mathcal{V}} \{ v \geq v_{i} \}$.
	%
	%
	\item The linear function $\ell: [0,1] \to [0,1]$ is such that $\ell(p)=p$ for every $p \in [0,1]$.
\end{itemize}
It is easy to verify that the monotonicity of the jumps is guaranteed, since $(\mu_i)_{i \in [n+1]}$ is indeed a decreasing sequence, as the cumulative density function $F$ is increasing by definition.

\subsection{No-Regret Learning in First-Price Auctions}\label{sec:applications_firstprice}

In a \emph{first-price} auction, a bidder with a private valuation $v \in [0,1]$ for the item being sold competes with other bidders in order to win the item, by placing a bid $b \in [0,1]$.
%
%
The bidders simultaneously send their bids, and the one with the highest bid wins the item, paying their own bid.
%
%
By letting $m \in [0,1]$ be the value of the highest bid submitted by the other bidders (\emph{i.e.}, excluding bid $b$), the bidder wins the item if and only if $b \geq m$ (assuming ties are broken in their favor), and they pay $b$, thus achieving a utility of $v - b$.
%
%
Otherwise, if $b < m$, the bidder does \emph{not} win, and their utility is equal to zero.
In the following, similarly to, \emph{e.g.}, \cite{han2024optimal}, we assume that the maximum bid submitted by the other bidders is drawn from a probability distribution $\mathcal{V}$, which is unknown to the bidder and satisfies $\text{supp}(\mathcal{V}) \subseteq [0,1]$ with $|\text{supp}(\mathcal{V})| = n$ for some $n \in \mathbb{N}_+$.
Thus, it is easy to verify that the bidder's expected utility, as a function of the submitted bid $b \in [0,1]$, is equal to $u(b) \coloneqq (v-b) \cdot \mathbb{P}_{m \sim \mathcal{V}} \{ b \geq m \}$.

When a bidder with valuation $v \in [0,1]$ repeats the auction for $T \in \mathbb{N}_+$ rounds, at each $t \in [T]$:\footnote{Notice that we assume that the bidder's valuation is the same at every round, while, in most of the works in the literature~\cite{han2024optimal,ai2022noregretlearningrepeatedfirstprice}, it is assumed that it is sampled at every round according to some fixed probability distribution. However, let us remark that, the latter setting cannot be captured by a BwMJ problem.}
%
%
\begin{enumerate}
	\item The bidder chooses a bid $b_t \in [0,1]$.
	%
	%
	\item The maximum bid of the other bidders $m_t \sim \mathcal{V}$ is sampled according to the distribution $\mathcal{V}$.
	\item The bidder’s utility is $v - b_t$ if $b_t \ge m_t$. Otherwise, if $b_t < m_t$, the bidder's utility is zero.
\end{enumerate}

The learning interaction described above corresponds to an equivalent one in a suitably-defined BwMJ instance.
We let each interval $\mathcal{A}_i$ be the set of bids for which the probability of winning the auction remains the same.
Furthermore, we let each distribution $\nu_i$ be a \emph{Bernoulli} distribution whose expected value corresponds to the probability of winning the auction in the interval $\mathcal{A}_i$.
Finally, we define the function $\ell(\alpha) = v - \alpha$.

%% file: tikz_plot/contract_reward.tikz
\begin{tikzpicture}
	\begin{axis}[
		name=plot1,
		axis lines=middle,
		xlabel={$\rho$},
		xlabel style={below right},
		ylabel={$u(\rho)$},
		ylabel style={left},
		xmin=0, xmax=1.1,
		ymin=0, ymax=0.3,
		xtick={0,0.5,0.8,1},
		ytick={0,0.1,0.2},
		grid=both,
		grid style={dashed,gray!50},
		extra y ticks={0.5, 0.8},
		extra y tick style={grid=major, grid style={red, thick, dashed}},
		width=7cm,
		height=6cm
		]
		
		\addplot[black, thick, domain=0:0.5, samples=100] {0.2*(1-x)};
		\addplot[black, thick, domain=0.5:0.8, samples=100] {0.5*(1-x)};
		\addplot[black, thick, domain=0.8:1, samples=100] {0.8*(1-x)};
	\end{axis}
	\begin{axis}[
		name=plot2,
		at={(plot1.right of south east)},
		anchor=left of south west,
		xshift=1.0cm,
		axis lines=middle,
		xlabel={$p$},
		xlabel style={below right},
		ylabel={$u(p)$},
		ylabel style={left},
		xmin=0, xmax=1.1,
		ymin=0, ymax=0.3,
		xtick={0,1/4,3/4,1},
		ytick={0,0.1,0.2},
		grid=both,
		grid style={dashed,gray!50},
		width=7cm,
		height=6cm
		]
		
		\addplot[black, thick, domain=0:1/4, samples=100] {x};
		\addplot[black, thick, domain=1/4:3/4, samples=100] {0.3*x};
		\addplot[black, thick, domain=3/4:1, samples=100] {0.0*x};
	\end{axis}
	\hspace{-0.8cm}
	
\end{tikzpicture}

%% file: no_regret.tex
\section{An Optimal No-Regret Algorithm for the BwMJ Framework}\label{sec:NoRegretAlgo}

In this section, we present the main result of this paper, which is a no-regret learning algorithm for BwMJ problems, whose regret guarantees are tight up to logarithmic factors.
The section is organized as follows.
Section~\ref{sec:MainAlgo} introduces the main procedure executed by the algorithm, while Section~\ref{sec:FindJumps} and Section~\ref{sec:OptShrink} focus on two core sub-procedures.
Finally, Section~\ref{sec:PutTogether} concludes by proving the regret guarantees attained by the algorithm.

\subsection{Recursive Jump Identification With Optimistic Shrinking}\label{sec:MainAlgo}

The main procedure executed by the algorithm---called \emph{Recursive Jump Identification with Optimistic Shrinking} (\texttt{RJI-OS} for short)---is provided in Algorithm~\ref{alg:main}.
%
%
The procedure splits the $T$ rounds into several epochs, with each epoch $j \in \mathbb{N}_{+}$ being further divided into three macro blocks.\footnote{Notice that, for the sake of presentation, in Algorithm~\ref{alg:main} and all its sub-procedures we do \emph{not} keep track of the current round $t \in [T]$. However, we assume that, whenever $t > T$, the execution of the algorithm is immediately stopped.}

In the \emph{first} macro block (Lines~\ref{line:firstformain}-\ref{line:endformain}), Algorithm~\ref{alg:main} attempts at identifying all the jump discontinuities whose corresponding gap $\mu_{i+1} -\mu_i$ is of the order of $\Delta_j \coloneqq \nicefrac{1}{2^j}$. 
To do this, the algorithm invokes the \texttt{Find-Jumps} procedure (Algorithm~\ref{alg:find_bp}), which employs a recursive binary-search-style approach.
In the \emph{second} macro block (Lines~\ref{line:beginSecondBlock}-\ref{line:endSecondBlock}), by leveraging the information about jumps collected by \texttt{Find-Jumps}, Algorithm~\ref{alg:main} computes an estimate of $\textnormal{OPT}$ for the current epoch $j$, denoted as $\textnormal{OPT}(j)$.
By using such an estimate and the information collected by \texttt{Find-Jumps}, in the \emph{third} macro block (Lines~\ref{line:beginThirdBlock}-\ref{line:endThirdBlock}), Algorithm~\ref{alg:main} reduces the set of actions $\alpha \in [0,1]$ to be considered during epoch $j+1$.
%
%
This is done by employing a suitably-designed procedure called \texttt{Optimistic-Shrink} (Algorithm~\ref{alg:restrict}), whose effects are crucial in order to achieve the desired regret guarantees.

Algorithm~\ref{alg:main} keeps track of all the information it needs into two data structures.
The first one is a collection of action intervals that the \texttt{Find-Jumps} procedure has to consider during the current epoch, denoted as $\mathcal{I}_j$ for epoch $j \in \mathbb{N}_+$.
Such a collection contains the whole action space $[0,1]$ during the first epoch, namely $\mathcal{I}_1 \coloneqq \{ [0,1] \}$, and it is updated at the end of each epoch (in the third macro block) by the \texttt{Optimistic-Shrink} procedure.
Notice that, during an epoch $j \in \mathbb{N}_+$, the set $\mathcal{I}_j$ may contain potentially non-adjacent action intervals, which do \emph{not} necessarily cover the whole action space $[0,1]$.
In the second data structure, which is denoted by $\mathcal{T}_j$ for epoch $j \in \mathbb{N}_+$, the algorithm keeps track of all the information collected by the \texttt{Find-Jumps} procedure in the current epoch.
Specifically, each element of $\mathcal{T}_j$ is a triplet specifying an action interval $I = [\alpha_1, \alpha_2]$ that has been returned by \texttt{Find-Jumps}, an estimate $\widehat \mu(\alpha_{1})$ of the expected value $\mu_{h(\alpha_1)}$, and an estimate $\widehat \mu(\alpha_{2})$ of $\mu_{h(\alpha_2)}$.
Notice that $\mathcal{T}_j$ is re-initialized at the beginning of each epoch.

\begin{algorithm}[!t]
	\caption{\texttt{Recursive Jump Identification with Optimistic Shrinking (RJI-OS)}}
	\label{alg:main}
	\begin{algorithmic}[1]
		\Require Number of rounds $T \in \mathbb{N}_+$
		\State Set $\mathcal{I}_1 \gets \{ [0,1] \}$
		\For{ $j \in \mathbb{N}_{+}$} 
		\State $\mathcal{T}_{j} \gets \varnothing$, $\Delta_j \gets \nicefrac{1}{2^j}$\label{line:firstformain}\Comment{\textcolor{gray}{First macro block}}
		\For{$I \in \mathcal{I}_{j}$}
		\State $\mathcal{T}_j\gets \mathcal{T}_j \cup \texttt{Find-Jumps}(I,\Delta_j ,1)$ 	\label{line:endformain}
		\EndFor
		\State $\text{OPT}(j) \gets 0$,  $\mathcal{I}_{j+1} \gets \varnothing$ \Comment{\textcolor{gray}{Second macro block}}\label{line:beginSecondBlock}
		\For{$ ( [\alpha_{1}, \alpha_{2}],\widehat \mu(\alpha_1), \widehat \mu(\alpha_2) ) \in \mathcal{T}_{j}$} \label{line:For} 
		\For{$\alpha \in \{\alpha_1,\alpha_2\}$}
		\If{$\ell(\alpha) \cdot \widehat \mu(\alpha) > \text{OPT}(j) $}
		\State $\text{OPT}(j) \gets \ell(\alpha) \cdot \widehat \mu(\alpha) $, $\alpha_j^\star \gets \alpha$
		\EndIf
		\EndFor 
		\EndFor \label{line:endSecondBlock}
		\For{$ \tau = ( [\alpha_{1}, \alpha_{2}],\widehat \mu(\alpha_1), \widehat \mu(\alpha_2) )\in \mathcal{T}_{j}$}  \Comment{\textcolor{gray}{Third macro block}}\label{line:beginThirdBlock}
		\State $\mathcal{I}_{j+1}\gets \mathcal{I}_{j+1} \cup \texttt{Optimistic-Shrink}(\tau,\text{OPT}(j),j)$ \label{line:end_first_algo}
		\EndFor
		\EndFor \label{line:endThirdBlock}
	\end{algorithmic}
\end{algorithm}

Next, we provide more details about the three macro blocks in each epoch $j \in \mathbb{N}_{+}$ of Algorithm~\ref{alg:main}. 
\begin{enumerate}
\item In the \emph{first} macro block (Lines~\ref{line:firstformain}-\ref{line:endformain}), Algorithm~\ref{alg:main} iterates over the intervals in $\mathcal{I}_{j}$.
For each interval, Algorithm~\ref{alg:main} invokes the \texttt{Find-Jumps} procedure to identify, within the interval, some jump discontinuity whose corresponding gap is of the order of $\Delta_j$.
To do so, \texttt{Find-Jumps} recursively calls itself, each time halving the interval given as input.
%
%
Specifically, given an interval $I = [\alpha_{1},\alpha_{2}]$, \texttt{Find-Jumps} first plays both the extremes of the interval for a number of rounds proportional to $\nicefrac{1}{\Delta_j^2}$.
%
This allows to build suitable estimates $\widehat \mu(\alpha_{1})$ and $\widehat \mu(\alpha_{2})$ of $\mu_{h(\alpha_{1})}$ and $\mu_{h(\alpha_{2})}$, respectively.
If the difference between such estimates is larger than $\Delta_j$, then \texttt{Find-Jumps} recursively invokes itself on the two halves of $I$.
%
Otherwise, it stops returning a triplet defined as $(I,\widehat \mu(\alpha_{1}),\widehat \mu(\alpha_{2}))$.
%
All the triplets generated by the different executions of \texttt{Find-Jumps} are stored by Algorithm~\ref{alg:main} in $\mathcal{T}_j$, in order to be used in the subsequent blocks.
%

\item In the \emph{second} macro block (Lines~\ref{line:beginSecondBlock}-\ref{line:endSecondBlock}), Algorithm~\ref{alg:main} computes the current estimate $\textnormal{OPT}(j)$ of the optimal expected utility $\textnormal{OPT}$, by using the estimates previously computed by \texttt{Find-Jumps}.
Specifically, Algorithm~\ref{alg:main} iterates over all the triplets $(I = [\alpha_{1},\alpha_{2}], \widehat \mu(\alpha_{1}), \widehat \mu(\alpha_{2})) \in \mathcal{T}_{j}$, and, for each of them, it computes the estimated expected utility at both the extremes of $I$, as $\ell(\alpha) \cdot \widehat \mu(\alpha)$ for $\alpha \in \{\alpha_{1},\alpha_{2}\}$.
The maximum among all such values is taken as $\textnormal{OPT}(j)$, while the extreme point in which it is attained is stored in $\alpha_j^\star$.
%
%

\item In the \emph{third} macro block (Lines~\ref{line:beginThirdBlock}-\ref{line:endThirdBlock}), Algorithm~\ref{alg:main} uses the \texttt{Optimistic-Shrink} procedure to identify a collection $\mathcal{I}_{j+1}$ of action intervals to be considered by the \texttt{Find-Jumps} procedure in the following epoch $j+1$.
The main goal of this step is to ensure that, in each possible $\alpha \in [0,1]$ selected during epoch $j+1$, the learner's expected utility satisfies $u(\alpha) \ge \textnormal{OPT} - \mathcal{O}(\Delta_j)$. 
In this way, the regret incurred by the algorithm during the execution of \texttt{Find-Jumps} in the first block of epoch $j+1$ is \emph{not} too large.
This is crucial to achieve the desired regret guarantees.
%
\end{enumerate}

\subsection{The \texttt{Find-Jumps} Procedure}\label{sec:FindJumps}


In this section, we present the \texttt{Find-Jumps} procedure, whose pseudocode is provided in Algorithm~\ref{alg:find_bp}.
At an high level, given an interval as input, the algorithm implements a recursive binary search in order to identify, within the interval, jump discontinuities whose associated gaps are proportional to a given $\Delta_j$.
Specifically, the algorithm takes as input a triplet $(I, \Delta_j, k)$, where $I = [\alpha_1, \alpha_2] \subseteq [0,1]$ is an action interval, $\Delta_j \in [0,1]$ is a jump gap considered by Algorithm~\ref{alg:main} in epoch $j \in \mathbb{N}_+$, and $k \in \mathbb{N}_{+}$ is the depth reached in the tree of recursive calls made by \texttt{Find-Jumps}.

Before delving into the details of Algorithm~\ref{alg:find_bp}, we make an observation that helps to better understand how it works.
%
%
Given an interval $I = [\alpha_{1},\alpha_{2}]\subseteq [0,1]$,
by playing each $\alpha_{i}$ for $\mathcal{O}\left(\nicefrac{1}{\Delta_j^2} \right)$ rounds, it is possible to compute some estimates $\widehat \mu(\alpha_{i})$ whose distance from the true expected values $\mu_{h(\alpha_{i})}$ is of the order of $\Delta_j$, with high probability.
%
%
Thus, if $\widehat \mu(\alpha_{2}) - \widehat \mu(\alpha_{1})$ is larger than $\Delta_j$, then, with high probability, there exists at least one jump discontinuity within the interval $I$.
However, notice that this does \emph{not} guarantee that there exists a \emph{single} jump discontinuity with gap proportional to $\Delta_j$. 
Indeed, there could be multiple jump discontinuities with smaller gaps, whose cumulative sum is proportional to $\Delta_j$.
Instead, if $\widehat \mu(\alpha_{2}) - \widehat \mu(\alpha_{1})$ is smaller than $\Delta_j$, then either there is no jump discontinuity in the interval $I$ or a gap smaller than $\Delta_{j}$ is needed in order to identify one.
%


\begin{algorithm}[!htp]
	\caption{\texttt{Find-Jumps}}
	\label{alg:find_bp}
	\begin{algorithmic}[1]
		\Require A triplet $(I,\Delta_j,k)$ with $I = [\alpha_1, \alpha_2] \subseteq [0,1]$, $\Delta_j \in [0,1]$, and $k\in \mathbb{N}_{+}$
		\State $\delta \gets \nicefrac{1}{T} $ 
		\State $N \gets \nicefrac{8}{\Delta_j^2} \log( \nicefrac{4T}{\delta} )$ 
		\If{$\alpha_2 - \alpha_1 \le \frac{1}{T}$} \label{line:small_int}
		\State Play $\alpha_t \coloneqq \alpha_{2}$ and observe feedback $x_t \sim \nu_{h(\alpha_t)}$ for $N$ rounds\label{line:small_int_ini}
		\State Compute estimate $\widehat \mu (\alpha_2)$ using the observed feedback
		\State $\widehat\mu(\alpha_1)  \gets 0$
		%
		%
		\State \textbf{return} $(I, \widehat\mu(\alpha_1), \widehat\mu(\alpha_2))$ \label{line:small_int_end}
		\EndIf
		\For{$i = 1,2 $} \label{line:collect_samples_findbp1}
		\State Play $\alpha_t \coloneqq \alpha_{i}$ and observe feedback $x_t \sim \nu_{h(\alpha_t)}$ for $N$ rounds\label{line:estimates}
		\State Compute estimate $\widehat \mu (\alpha_i)$ using the observed feedback
		%
		%
		\EndFor \label{line:collect_samples_findbp2}
		\If{$\widehat{\mu}(\alpha_{2}) -\widehat{\mu}(\alpha_{1}) \ge  \Delta_j $} \label{line:large_estimates}
		\State 	\textbf{return} $ \texttt{Find-Jumps}\left( [\alpha_1, \frac{\alpha_1+\alpha_2}{2}], \Delta_j, k+1 \right) \cup  \texttt{Find-Jumps} \left( [ \frac{\alpha_1+\alpha_2}{2}, \alpha_2], \Delta_j, k+1 \right)$ \label{line:recursion}
		\Else \label{line:small_estimates}
		\State \textbf{return} $(I, \widehat\mu(\alpha_1), \widehat\mu(\alpha_2))$\label{line:small_estimates_end}
		\EndIf
	\end{algorithmic}
\end{algorithm}

We now describe how Algorithm~\ref{alg:find_bp} works.
This can be conceptually split into two cases, depending on whether the length of the interval $I = [\alpha_{1},\alpha_{2}]$ given as input is greater than $\nicefrac 1 T$ or not.

\paragraph{Case $\alpha_{2} - \alpha_{1} \leq \nicefrac 1 T$}
%
%
In this case (see Lines~\ref{line:small_int}-\ref{line:small_int_end}), Algorithm~\ref{alg:find_bp} does \emph{not} search for jump discontinuities within the interval in input, and it returns the interval itself without further halving it.
Intuitively, this is because the interval is sufficiently small for the purposes of the algorithm.
This represents the first base case of the recursive binary search implemented by Algorithm~\ref{alg:find_bp}.
Before returning, the algorithm plays for $N \propto \nicefrac{1}{\Delta_j^2}$ rounds the right extreme $\alpha_2$ of the interval, in order to compute an estimate $\widehat \mu(\alpha_{2})$ of $\mu_{h(\alpha_2)}$, while it artificially sets the estimate $\widehat \mu(\alpha_{1})$ for the left extreme $\alpha_{1}$ to zero.
%
%
%
%
The reason why Algorithm~\ref{alg:find_bp} computes an estimate only for the right extreme of the interval is that, as it will become clear later on, this extreme provides an expected utility that is sufficiently close to the highest expected utility attainable over the whole interval.
Finally, in this case the algorithm ends its execution with Line~\ref{line:small_int_end}, by returning the triplet $(I, \widehat\mu(\alpha_1), \widehat\mu(\alpha_2))$.
%
%
%
%
%
%

\paragraph{Case $\alpha_{2} - \alpha_{1} > \nicefrac 1 T$}
%
%
In this case (see Lines~\ref{line:collect_samples_findbp1}-\ref{line:small_estimates_end}), Algorithm~\ref{alg:find_bp} searches for a jump discontinuity within the interval given as input.
In order to do so, the algorithm first computes two estimates $\widehat \mu(\alpha_1)$ and $\widehat \mu(\alpha_2)$ of the expected values $\mu_{h(\alpha_1)}$ and $\mu_{h(\alpha_2)}$ at the extremes of the interval, by playing each $\alpha_i$ for a suitable number of rounds $N$ proportional to $\nicefrac{1}{\Delta_j^2}$.
%
%
%
Then, depending on the value of the difference between such estimates, the algorithm proceeds its execution as follows.
\begin{enumerate}
\item If the difference $\widehat \mu(\alpha_2) - \widehat \mu(\alpha_1)$ is smaller than $\Delta_j$, Algorithm~\ref{alg:find_bp} stops searching for jumps discontinuities within the input interval, and, thus, it does \emph{not} invoke itself anymore (see Line~\ref{line:small_estimates_end}).
Indeed, in this case, either there is no jump discontinuity in the interval or a smaller value of $\Delta_j$ is needed in order to identify one, as we previously observed.
%
%
This represents another base case in the recursive binary search implemented by Algorithm~\ref{alg:find_bp}, in which the algorithm simply returns the triplet defined as $(I, \widehat\mu(\alpha_1), \widehat\mu(\alpha_2))$ (see Line~\ref{line:small_estimates_end}). 
\item If the difference $\widehat \mu(\alpha_2) - \widehat \mu(\alpha_1)$ is larger than $\Delta_j$, Algorithm~\ref{alg:find_bp} halves the input interval $I$ and recursively invokes itself with these two new subintervals as inputs (see Line~\ref{line:large_estimates}).
This is because, if the difference between the estimates is greater than $\Delta_j $, then there could be a jump discontinuity whose gap is proportional to $\Delta_j$ within that interval, as previously observed.    
\end{enumerate}
%
%
%
%

Next, we provide some results related to Algorithm~\ref{alg:find_bp}.
In the following, for the ease of presentation, we introduce a \emph{clean event} $\mathcal{E}$.
This encodes all the situations in which the estimates computed by Algorithm~\ref{alg:find_bp} fall within a suitable confidence interval, in \emph{every} execution of the algorithm.
%
%
\begin{definition}[Clean event]
	 We define $\mathcal{E}$ as the event in which, every time Algorithm~\ref{alg:find_bp} is called given as input a triplet $(I, \Delta_j, k)$ with $I = [\alpha_1, \alpha_2] \subseteq [0,1]$, $\Delta_j \in [0,1]$, and $k \in \mathbb{N}_{+}$, the estimates $\widehat \mu(\alpha_1)$ and $\widehat \mu(\alpha_2)$ computed by the algorithm satisfy the following conditions:
	 %
	 	\begin{equation*}
	 		\left| \widehat \mu(\alpha_i)- \mu_{h(\alpha_i)} \right| \le \frac{\Delta_j}{4} \quad  \forall i \in \{1,2\}.
	 	\end{equation*}
\end{definition}   
    
We now prove two lemmas that state two crucial properties that Algorithm~\ref{alg:find_bp} satisfies.
The first property ensures that, whenever Algorithm~\ref{alg:find_bp} stops recursively calling itself, then either the interval given as input is ``sufficiently'' tight or the jump discontinuities within the interval (if there are some) have gaps that are ``small enough''.
%
%
%
Formally:
\begin{restatable}{lemma}{FindBPFirst}\label{lem:find_bp_delta}
	%
	Suppose that Algorithm~\ref{alg:find_bp} is given as input a triplet $(I,\Delta_j,k)$, with $I = [\alpha_1, \alpha_2] \subseteq [0,1]$, $\Delta_j \in [0,1]$, and $k \in \mathbb{N}_{+}$.
	Under the clean event $\mathcal{E}$, if the algorithm does not recursively call itself, then one of the following two conditions holds:
	\begin{enumerate}
		\item $\alpha_2 -\alpha_1 \le \frac{1}{T} $.
		\item $\mu_{h(\alpha)} - \mu_{h(\alpha')} \le \frac{3}{2} \Delta_j$ for all $ \alpha, \alpha' \in I$ such that $\alpha \ge \alpha'$.
	\end{enumerate}
	%
\end{restatable}
\begin{proof}
	To prove the lemma, we notice that when Algorithm~\ref{alg:find_bp} returns a tuple containing the same interval $I$ it has received as input, either the condition $|I| \le \nicefrac{1}{T}$ or the condition $\widehat{\mu}(\alpha_2) - \widehat{\mu}(\alpha_1) <  \Delta_j$ holds.
	This is due to the termination conditions specified in Algorithm~\ref{alg:find_bp}.
	
	In the following, we consider these two possible termination conditions and we show that if either one of them holds, then either $(1)$ or $(2)$ holds.
	\begin{enumerate}
		\item 	If $|I| \le \nicefrac{1}{T}$ , then condition (1) trivially holds.
		\item 	If $\widehat \mu (\alpha_2) -\widehat \mu (\alpha_1) < \Delta_j$, then 
		for all $\alpha, \alpha' \in I$ with $\alpha \geq \alpha'$, under the event $\mathcal{E}$, we have:
		\begin{equation*}
			\mu_{h(\alpha)}-\mu_{h(\alpha')}  \le \mu_{h(\alpha_2)}- \mu_{h(\alpha_1)} \le \widehat \mu (\alpha_2) -\widehat \mu (\alpha_1) +  \frac{\Delta_j}{2} \le  \frac{3}{2} \Delta_j.
		\end{equation*}
		Where the first inequality holds because of monotonicity, the second inequality holds under the event $\mathcal{E}$ and the last inequality holds because $\widehat \mu (\alpha_2) -\widehat \mu (\alpha_1)  < \Delta_j$ by hypothesis, and, thus condition (2) holds.
	\end{enumerate}
	Thanks to the two points above, the proof is concluded.
\end{proof}

The second crucial property that Algorithm~\ref{alg:find_bp} ensures is that, whenever it receives as input an interval such that the distributions at its extremes coincide, and, thus, they have the same expected value, the algorithm does \emph{not} invoke itself further.
Such a property is crucial; otherwise, Algorithm~\ref{alg:find_bp} may perform too many unnecessary recursive calls, resulting in a large regret. 
%
%
Formally:
\begin{restatable}{lemma}{FindBPSecond}\label{lem:find_bp_terminate}
	Suppose that Algorithm~\ref{alg:find_bp} is given as input a triplet $(I,\Delta_j,k)$ with $I=[\alpha_1,\alpha_2] \subseteq [0,1]$, $\Delta_j \in [0,1]$, and $k \in \mathbb{N}_{+}$ such that $\mu_{h(\alpha_1)} = \mu_{h(\alpha_2)}$.
	Then, under the clean event $\mathcal{E}$, the algorithm stops recursively invoking itself and terminates.
\end{restatable}
\begin{proof}
	We split the proof into two different cases.
	\begin{enumerate}
		\item If $|I| \le \nicefrac{1}{T}$, then the condition at Line~\ref{line:small_int} is satisfied and thus the Algorithm~\ref{alg:find_bp} returns the pair $(I, j, k)$, independently of the distributions at the extremes.
		
		\item If $|I| > \nicefrac{1}{T}$, under the event $\mathcal{E}$, we have:
		\begin{equation*}
			\widehat \mu (\alpha_2) -\widehat \mu (\alpha_1) \le  \mu_{h(\alpha_2)} - \mu_{h(\alpha_1)} + \frac{\Delta_j}{2} = \frac{\Delta_j}{2}.
		\end{equation*}
		Thus, the condition at Lines~\ref{line:large_estimates} is not satisfied and therefore the algorithm returns the tuple $(I, j, k)$, concluding the proof.
	\end{enumerate}
\end{proof}

\subsection{The \texttt{Optimistic-Shrink} Procedure}\label{sec:OptShrink}

We now introduce the second key procedure employed by Algorithm~\ref{alg:main}, which is \texttt{Optimistic-Shrink}, whose pseudocode is provided in Algorithm~\ref{alg:restrict}.
At a high level, the algorithm takes as input an interval $I \subseteq [0,1]$ and returns a subinterval $I' \subseteq I$ such that, for every action $\alpha \in I'$ selected during the following epoch $j+1$, it holds that $u(\alpha) \ge \textnormal{OPT} - \mathcal{O}(\Delta_j)$.
To do so, Algorithm~\ref{alg:restrict} relies on the estimates of the expected values at the extremes of the interval $I$, computed by Algorithm~\ref{alg:find_bp}.
%
%
As a result, Algorithm~\ref{alg:restrict} does \emph{not} require the learner to interact with environment, as it just manipulates the intervals in $\mathcal{I}_j$.
Thus, the regret suffered by the algorithm is zero.

Depending on the length of the interval in input, Algorithm~\ref{alg:restrict} works in two different ways.
\begin{enumerate}
	\item If the length of the input interval $I = [\alpha_1, \alpha_2]$ is such that $\alpha_2 - \alpha_1 > \nicefrac{1}{T}$, then Algorithm~\ref{alg:restrict} returns a subinterval $I' \subseteq I$ such that, if $\alpha \in I'$, then the value of $\ell(\alpha)\widehat{\mu}(\alpha_1) + 2\Delta_j + \nicefrac{2}{T} $ is larger than $\text{OPT}(j)$ (see Line~\ref{line:restrict_large_int}).
	This is because, thanks to Lemma~\ref{lem:find_bp_delta} and the definition of the clean event $\mathcal{E}$, the value $\ell(\alpha)\widehat{\mu}(\alpha_1) + 2 \Delta_j + \nicefrac{2}{T}$ represents an optimistic estimate of the expected utility in each $\alpha \in I$.
	Therefore, if such an optimistic estimate is smaller than $\text{OPT}(j)$, then the learner's expected utility will be suboptimal at that given $\alpha$, as a consequence of the relationship between $\text{OPT}(j)$ and $\text{OPT}$, which we show in the following Lemma~\ref{lem:optimal_estimate}.
	%
	%
	%
	%
	\item If the length of the input interval $I= [\alpha_1, \alpha_2]$ is such that $\alpha_2 -\alpha_1 \le \nicefrac{1}{T}$, then Algorithm~\ref{alg:restrict} returns the interval itself only if the optimistic estimate of the learner's expected utility at the right extreme $\alpha_{2}$ is larger than the current estimate of the optimum $\text{OPT}(j)$ (see Line~\ref{line:restrict_small_int}).
	This is because such an optimistic estimate, due to the monotonicity of the jumps and the fact that the input interval is sufficiently small, represents an upper bound for the learner's expected utility over the entire interval.
	Hence, if the estimate is smaller than $\text{OPT}(j)$, then the learner's expected utility within that interval is suboptimal, given the relationship between $\text{OPT}(j)$ and $\text{OPT}$ in Lemma~\ref{lem:optimal_estimate}.
	As a result, Algorithm~\ref{alg:restrict} does \emph{not} return the interval $I$.
\end{enumerate}
%
%
%
%

Next, we prove some crucial properties that Algorithm~\ref{alg:restrict} satisfies.
The first property ensures that Algorithm~\ref{alg:restrict} always returns an interval $I \subseteq[0,1]$ containing an optimal action $\alpha^\star$.
%

Formally:
\begin{restatable}{lemma}{RestrictFirst}\label{lem:always_optimal}
	Under the clean event $\mathcal{E}$, for every epoch $j \in \mathbb{N}_{+}$, it holds $\alpha^\star \in I$ for some $I \in \mathcal{I}_j$.
\end{restatable}

\begin{proof}
	To prove the lemma we employ an inductive argument. When $j=1$, the statement is trivially satisfied since $\mathcal{I}_1= \{[0,1]\}$ and $\alpha^\star \in [0,1] $.
	
	Suppose that $\alpha^\star \in I$ for some $I \in  \mathcal{I}_{j}$ with $j > 1$, we prove that $\alpha^\star \in I $ for some $I \in  \mathcal{I}_{j+1}$. To show that, we consider two different cases.
	\begin{enumerate}
		\item  If $\alpha^\star \in I$ for some $I \in \mathcal{I}_j$ with $|I| > \nicefrac{1}{T}$, then it holds:
		\begin{align*}
			\ell(\alpha^\star) \widehat \mu(\alpha_1) 
			& \ge \ell(\alpha^\star)  \mu_{h(\alpha_1)}  -\frac{\Delta_j}{4} \\
			& \ge \ell(\alpha^\star) \mu_{h(\alpha^\star)} -\frac{7 \Delta_j}{4}\\
			& = \text{OPT} -\frac{7 \Delta_j}{4}\\
			& \ge \ell(\alpha_j^\star) \mu_{h(\alpha_j^\star)}  -\frac{7 \Delta_j}{4}\\
			& \ge \ell(\alpha_j^\star) \widehat \mu(\alpha_j^\star)  -{2 \Delta_j} \\
			& = \text{OPT}(j) -{2 \Delta_j}.
		\end{align*}

		Where the first inequality above holds under the clean event $\mathcal{E}$.
		The second inequality holds thanks to Lemma~\ref{lem:find_bp_delta} since $|I| \ge \nicefrac{1}{T}$, while the third inequality holds because of the optimality of \textnormal{OPT}.
		The fourth inequality holds under the event $\mathcal{E}$, and the last inequality holds because of the definition of $\alpha_j^\star \in [0,1]$.
		
		Thanks to the above inequalities, we have the following:
		\begin{equation*}
			\ell(\alpha^\star) \widehat \mu(\alpha_1)  + 2\Delta_j \ge \textnormal{OPT}(j).
		\end{equation*}
		Consequently, the condition at Line~\ref{line:restrict_large_int} is satisfied, and thus $\alpha^{\star}$ also belongs to some $I \in \mathcal{I}_{j+1}$.
		
		\item 	If $\alpha^\star \in I$ for some $I \in \mathcal{I}_j$ with $|I| \le \nicefrac{1}{T}$, then the following holds:
		\begin{align*}
			\ell ( \alpha_2)\widehat \mu(\alpha_2)  
			& \ge \ell(\alpha_2 ) \mu_{h(\alpha_2  )} -\frac{\Delta_j}{4} \\
			& \ge \ell(\alpha_2 ) \mu_{h(\alpha^\star  )} -\frac{\Delta_j}{4}\\
			& =( \ell(\alpha_2) \pm \ell( \alpha^{\star} )) \mu_{h(\alpha^\star  )} -\frac{\Delta_j}{4}\\
			& \ge \ell(\alpha^\star )  \mu_{h(\alpha^\star  )} -\frac{\Delta_j}{4} - \frac{1}{T} \\
			& = \text{OPT} -\frac{\Delta_j}{4} - \frac{1}{T} \\
			& \ge \ell(\alpha_j^\star )   \mu_{h(\alpha_j^\star  )} -\frac{\Delta_j}{4}- \frac{1}{T}\\
			& \ge \ell(\alpha_j^\star )  \widehat \mu(\alpha_j^\star  ) -\frac{\Delta_j}{2} - \frac{1}{T} \\
			& = \text{OPT}(j) - \frac{\Delta_j}{2} - \frac{1}{T}
		\end{align*}
		We notice that the first and last inequalities above hold under the clean event $\mathcal{E}$. The second inequality holds since $\mu_{h(\alpha_2)} \ge \mu_{h(\alpha^\star)}$ according to monotonicity. The third inequality holds by hypothesis observing that $|\alpha^\star - \alpha_1| \le |I| \le \nicefrac{1}{T}.$
		As a result, the following holds:
		\begin{equation*}
			\ell(\alpha_2) \widehat{\mu}(\alpha_2) + \Delta_j/2 + 1/T \ge \textnormal{OPT}(j).
		\end{equation*}
		Consequently, the condition at Lines~\ref{line:restrict_small_int} is satisfied, and thus $\alpha^{\star}$ also belongs to some $I \in \mathcal{I}_{j+1}$.
	\end{enumerate}
	Finally, since both the base case holds and the inductive step is valid, we conclude the proof.
\end{proof}
\begin{algorithm}[h]
	\caption{\texttt{Optimistic-Shrink}}
	\label{alg:restrict}
	\begin{algorithmic}[1]
		\Require $(I=[\alpha_1, \alpha_2],\widehat{\mu}(\alpha_1), \widehat{\mu}(\alpha_2)) \in \mathcal{T}_j$, $\Delta_j \in [0,1]$, and $\text{OPT}(j) \in [0,1]$
		\If{$\alpha_2 - \alpha_1 \le \frac{1}{T} \, \wedge \,  \ell(\alpha_2) \cdot \widehat{\mu}(\alpha_2) + \nicefrac{\Delta_j}{2}+ \nicefrac{1}{T}\ge \textnormal{OPT}(j)$}  \label{line:restrict_small_int}
		\State\Return $I$
		\EndIf
		\If{$\alpha_2 - \alpha_1 > \frac{1}{T}$} 
		\State $ I' \gets \left\{ \alpha \in I \mid \ell(\alpha) \cdot \widehat{\mu}(\alpha_1) + 2 \Delta_j + \nicefrac{2}{T}\ge   \textnormal{OPT}(j) \right\}  $  \label{line:restrict_large_int}
		\State\Return $I'$ 
		\EndIf
		\State\Return $\varnothing$ 
	\end{algorithmic}
\end{algorithm}

In order to prove the lemma, we consider two different cases.
The first one is when $\alpha^\star \in I$ for some $I = [\alpha_{1},\alpha_{2}] \in \mathcal{I}_j$ with $\alpha_{2}-\alpha_{1}<\nicefrac{1}{T}$.
Then, under the clean event $\mathcal{E}$ and by the monotonicity of jumps, it is possible to show that the optimistic estimate of the learner's expected utility at the right extreme $\alpha_{2}$ of the interval satisfies the condition in Line~\ref{line:restrict_small_int} of Algorithm~\ref{alg:restrict}. 
Thus, the interval that contains $\alpha^\star$ is guaranteed to also belong to $\mathcal{I}_{j+1}$.
The second case is when $\alpha^\star \in I$ for some $I \in \mathcal{I}_j$ whose length is larger than $\nicefrac{1}{T}$.
Then, by employing Lemma~\ref{lem:find_bp_delta}, it is possible to show that the optimistic estimate of the learner's expected utility for action $\alpha^\star$ (evaluated in Line~\ref{line:restrict_large_int} of Algorithm~\ref{alg:restrict}) is greater than or equal to $\textnormal{OPT}(j)$.
Thus, if there exists an interval $I \in \mathcal{I}_j$ such that $\alpha^\star \in I$, then there is also some $I' \in \mathcal{I}_{j+1}$ such that $\alpha^\star \in I'$.
Finally, by using an inductive argument, the statement of the lemma is immediately proved.

Thanks to Lemma~\ref{lem:always_optimal}, it is possible to show that the difference between the optimal value $\text{OPT}$ and the one estimated at epoch $j \in \mathbb{N}_{+}$, namely $\text{OPT}(j)$, is proportional to $\Delta_j$. 
Formally, we have:
\begin{restatable}{lemma}{RestrictSecond}\label{lem:optimal_estimate}
	Under the clean event $\mathcal{E}$, it holds $\textnormal{OPT}(j) \ge \textnormal{OPT} - \nicefrac{ 7 \Delta_j }{4}-\nicefrac{1 }{T}$ for every epoch $j \in \mathbb{N}_+$.
\end{restatable}
\begin{proof}
	We split the proof into two cases.
	\begin{enumerate}
		\item If $\alpha^\star \in I$ with $(I, \widehat\mu(\alpha_1), \widehat\mu(\alpha_2) )\in \mathcal{T}_{j}$ and $|I| > \nicefrac{1}{T}$, then we have:
		\begin{align*}
			\textnormal{OPT}(j) &\ge \ell(\alpha_1) \widehat{\mu}(\alpha_1) \\
			&\ge  \ell(\alpha_1)  {\mu}_{h(\alpha_1)} - \frac{\Delta_j}{4}\\
			&\ge  \ell(\alpha_1)  {\mu}_{h(\alpha^\star)} - \frac{7 \Delta_j}{4} \\
			&\ge  \ell(\alpha^\star) {\mu}_{h(\alpha^\star)} - \frac{7 \Delta_j}{4}\\
			&= \textnormal{OPT} - \frac{7 \Delta_j}{4}.
		\end{align*}
		The first inequality above holds according to the definition of $\textnormal{OPT}(j)$, the second inequality holds under the event $\mathcal{E}$, the third inequality holds thanks to Lemma~\ref{lem:find_bp_delta} and the last inequality holds since $\alpha^\star \ge \alpha_1$.
		\item  If $\alpha^\star \in I$ with $(I, \widehat\mu(\alpha_1), \widehat\mu(\alpha_2) )\in \mathcal{T}_{j}$ and $|I| \le \nicefrac{1}{T}$, under the event $\mathcal{E}$, we have:
		\begin{align*}
			\textnormal{OPT}(j)
			& \ge \ell(\alpha_2)  \widehat{\mu}({\alpha_2})\\
			& \ge \ell(\alpha_2)  {\mu}_{h(\alpha_2)} - \frac{\Delta_j}{4}\\
			& = (\ell(\alpha_2 )\pm \ell(\alpha^\star) ) {\mu}_{h(\alpha^\star)} - \frac{\Delta_j}{4} \\
			& \ge \ell(\alpha^\star)  {\mu}_{h(\alpha^\star)}  - \frac{ \Delta_j}{4} - \frac{1}{T} \\
			& = \textnormal{OPT} - \frac{ \Delta_j}{4} - \frac{1}{T}.
		\end{align*}
		The first inequality holds according to the definition of $\textnormal{OPT}(j)$, the second inequality holds under the event $\mathcal{E}$, the third inequality holds thanks to monotonicity and the last inequality holds since $|I| \le \nicefrac{1}{T}$.
	\end{enumerate}
	Thanks to the above inequalities the lemma holds.
\end{proof}

%
%
%

Finally, we introduce the main property satisfied by Algorithm~\ref{alg:restrict}, which establishes a lower bound on the learner's expected utility for the intervals returned by the algorithm.
Formally:
\begin{restatable}{lemma}{RestrictThird}\label{lem:deltaopt}
	Under the event $\mathcal{E}$, for every epoch $j \in \mathbb{N}_+$ and interval $I=[\alpha_1, \alpha_2] \in \mathcal{I}_{j+1}$, it holds:
	\begin{enumerate}
		\item If $\alpha_{2}-\alpha_{1} > \nicefrac{1}{T}$, then $u(\alpha) \ge \textnormal{OPT} - 4\Delta_j - \nicefrac{2}{T}$ for every action $\alpha \in I$.
		\item If $\alpha_{2}-\alpha_{1} \leq \nicefrac{1}{T}$, then $u(\alpha_2) \ge \textnormal{OPT} - 4\Delta_j - \nicefrac{2}{T}$.
	\end{enumerate}
	%
\end{restatable}
\begin{proof}
	We split the proof into two cases.
	\begin{enumerate}
		\item If $|I| > \nicefrac{1}{T}$, then under the event $\mathcal{E}$, the following holds:
		\begin{align*}
			\ell(\alpha) \mu_{h(\alpha)}  &\ge \ell(\alpha) \mu_{h(\alpha_1)}  \\
			&\ge  \ell(\alpha) \widehat{\mu}(\alpha_1 ) - \frac{\Delta_j}{4} \\
			&\ge \textnormal{OPT}(j)- \frac{\Delta_j}{4} - {2 \Delta_j} - \frac{2}{T} \\
			&\ge  \textnormal{OPT}- {4 \Delta_j} - \frac{2}{T}.
		\end{align*}
		We notice that the first inequality above holds as a consequence of monotonicity, and the second inequality holds under the clean event $\mathcal{E}$. Furthermore, if $\alpha \in I$ for some $I \in \mathcal{I}_{j+1}$, then the condition at Lines~\ref{line:restrict_large_int} is satisfied, and thus the third inequality above holds. Finally, the last inequality holds as a consequence of Lemma~\ref{lem:optimal_estimate}.
		\item If $|I| \le \nicefrac{1}{T}$, we have:
		\begin{align*}
			\ell(\alpha_2) \mu_{h(\alpha_2)}  &\ge \ell(\alpha_2) \widehat\mu(\alpha_2) - \frac{\Delta_j}{4}  \\
			&\ge \textnormal{OPT}(j)- \frac{\Delta_j}{4} - \frac{\Delta_j}{2} - \frac{1}{T}  \\
			&\ge  \textnormal{OPT}- {4 \Delta_j} - \frac{2}{T}.
		\end{align*}
		The first inequality above holds under the clean event $\mathcal{E}$. The second inequality holds true because  the condition at Lines~\ref{line:restrict_small_int} is satisfied. Finally, the last inequality holds as a consequence of Lemma~\ref{lem:optimal_estimate}. 
	\end{enumerate}
    This concludes the proof.
\end{proof}

Notice that, thanks to Lemma~\ref{lem:deltaopt}, every action $\alpha \in \mathcal{I}_{j+1}$ chosen by the algorithm at epoch $j+1$ is $\mathcal{O}(\Delta_j)$-optimal.
Thus, since the \texttt{Find-Jumps} procedure employs a number of samples proportional to $\nicefrac{1}{\Delta_{j+1}^2}$, the resulting regret suffered by the algorithm at epoch $j+1$ is proportional to $\mathcal{O}\left( \nicefrac{\Delta_{j}}{\Delta_{j+1}^2}\right) = \mathcal{O}\left( \nicefrac{1}{\Delta_{j+1}}\right) $, since $2\Delta_{j+1} = \Delta_j$.
%

\subsection{Putting All Together}\label{sec:PutTogether}

In order to derive the regret guarantees of Algorithm~\ref{alg:main}, we first need to bound the number of calls made to \texttt{Find-Jumps} (including recursive ones) during each epoch $j \in \mathbb{N}_{+}$, since each call may require the learner to select two actions for a number of rounds proportional to $\nicefrac{1}{\Delta^2_j}$. Formally:
\begin{restatable}{lemma}{MainAlgoFirst}\label{lem:call_find_bp}
	Under the clean event $\mathcal{E}$, at each epoch $j \in \mathbb{N}_{+}$ of Algorithm~\ref{alg:main}, the number of calls to \texttt{Find-Jumps} is upper bounded by~$(j+2)n \lceil \log_2(T) \rceil$.
\end{restatable}
\begin{proof}
	In the following, we introduce $ \mathcal{N}(j,k) $, which corresponds to the set of tuples given as input to Algorithm~\ref{alg:find_bp} at iteration $ j \in \mathbb{N}_{+} $ and level $ k \in \mathbb{N}_{+} $.
	Formally, given $ j, k \in \mathbb{N}_{+} $, we define
	\begin{equation*}
		\mathcal{N}(j,k) \coloneqq |\{ I \subseteq [0,1]  \,\ | \,\ (I,j,k) \,\, \textnormal{is given as input to Algorithm~\ref{alg:find_bp}}  \}|
	\end{equation*}
	Furthermore, we denote $N(j,k) \coloneqq |	\mathcal{N}(j,k) |$. Intuitively, the quantity $N(j,k)$ corresponds to the number of nodes at level $k \in \mathbb{N}_{+}$ in the recursive tree generated by the execution of    Algorithm~\ref{alg:find_bp}. (See Figure~\ref{fig:tree}).
	As a result, the quantity $ \sum_{k \in \mathbb{N}_{+}} N(j,k) $ represents the number of calls performed to Algorithm~\ref{alg:find_bp} during the execution of the for-loop at Lines~\ref{line:firstformain} in Algorithm~\ref{alg:main} at iteration $j \in \mathbb{N}_{+}$.
	
	We also observe that, when Algorithm~\ref{alg:find_bp} receives a tuple as input, it either recursively invokes itself twice or it does not recursively invoke itself anymore.
	Thus, each node in the recursive tree in Figure~\ref{fig:tree} has either zero or two children.
	This means that the quantity $N(j, k+1)/2$ represents the number of times the algorithm has recursively invoked itself at level $k \in \mathbb{N}_{+}$ and iteration $j \in \mathbb{N}_{+}$.
	Consequently, it represents the number of nodes that have two children at level $k \in \mathbb{N}_{+}$ in the tree of  Figure~\ref{fig:tree}.
	
	We also introduce the set $\mathcal{M}(j,k)$, which contains all the tuples corresponding to the leaves of the recursive tree generated by the execution of Algorithm~\ref{alg:find_bp} (see the red boxes in Figure~\ref{fig:tree}) at level $k \in \mathbb{N}_{+}$ and iteration $j \in \mathbb{N}_{+}$. Formally, we have:
    \begin{align*}\small
    \mathcal{M}(j,k) \hspace{-1mm}\coloneqq \mathcal{N}(j,k) \cup \{ I \subseteq [0,1], k' \in [k-1] \,|\, \textnormal{$(I, j, k')$ is an input of Algo~\ref{alg:find_bp} s.t. it does not invoke itself} \}.
	\end{align*}

	We denote $M(j,k) \coloneqq |	\mathcal{M}(j,k) |$ and we observe that can be recursively expressed as follows:
	\begin{equation}\label{eq:recursion_m}
		M(j,k+1) = \underbrace{M(j,k) - \frac{N(j,k+1)}{2} }_{(1)}+  2\underbrace{\frac{N(j,k+1)}{2}}_{(2)} = M(j,k) + \frac{N(j,k+1)}{2},
	\end{equation}
	for all $k\ge1$, with $M(j,1) = |\mathcal{I}_{j}|$.
	As observed before, the quantity (2) in Equation~\ref{eq:recursion_m} represents the number of nodes at level $k$ that have two children and thus will be leaf nodes at level $k+1$.
	The quantity in (1) in Equation~\ref{eq:recursion_m} represents all the leaf nodes at level $k$ that have no children and thus will remain leaves at level $k+1$ as well.
	
	Furthermore, for each $j \in \mathbb{N}_{+}$, the following holds.
	\begin{enumerate}
		\item If $k=1$, then $N(j,1)=|\mathcal{I}_{j}|$. This is because Algorithm~\ref{alg:find_bp} receives as input all the tuples $(I,j,1)$ with $I \in \mathcal{I}_j$, as specified at  Lines~\ref{line:firstformain} in Algorithm~\ref{alg:main}. 
		
		\item If $k> \lceil \log_2(T) \rceil + 1 $, then $N(j,k)=0$.
		When Algorithm~\ref{alg:find_bp} receives a tuple $(I, j, k)$ as input and it recursively invokes itself with a new tuple $(I', j, k + 1)$, the length of $I'$ is halved compared to that of $I$.
		Thus, observing that Algorithm~\ref{alg:find_bp} does not invoke itself again if it receives an interval $I \subseteq [0, 1]$ that satisfies $ |I| \le \nicefrac{1}{T}$, we conclude that there is no interval $ I $ such that the tuple $ (I, \lceil \log_2(T) \rceil + 2, j) $ may represent a valid input to Algorithm~\ref{alg:find_bp}.
		Indeed, any interval $I \subseteq [0, 1]$, when halved $\lceil \log_2(T) \rceil $ times, will ultimately have a length smaller than $ \nicefrac{1}{T} $.
		
		\item If $ 0 < k \le \lceil \log_2(T) \rceil +1 $, then $ N(j, k) \le 2n $.
		This is because when Algorithm~\ref{alg:find_bp} receives an interval $I=[\alpha_1, \alpha_2]$ as input, it recursively invokes itself only if the condition $\mu_{h(\alpha_1)} < \mu_{h(\alpha_2)}$ holds, as stated in Lemma~\ref{lem:find_bp_terminate}, and thus only if there is at least one jump discontinuity within the interval.
		Furthermore, observing that there are at most $n-1$ jumps and that Algorithm~\ref{alg:find_bp} recursively invokes itself twice, we have $ N(j, k) \le 2(n-1) \le 2n $.
	\end{enumerate}
	Then, thanks to Equation~\ref{eq:recursion_m} and the above observations, we have that:
	\begin{equation*}
		M(j,k)  = \begin{cases}
			M(j,k+1) \le M(j,k) \,\,\,\, \forall \, k> \lceil \log_2(T) \rceil, \\
			M(j,k+1) \le M(j,k)+ n \,\,\,\, \forall \, 1\le k\le \lceil \log_2(T)\rceil, \\
			M(j,1) = |\mathcal{I}_{j}|.
		\end{cases}
	\end{equation*}
	Consequently, at each iteration $j \in \mathbb{N}_{+}$, we have:
	\begin{equation*}
		M(j,k) \le n \lceil \log_2(T) \rceil+   |\mathcal{I}_{j}|,
	\end{equation*}
	for all $k >0$.

	We also observe that $|\mathcal{T}_{j}| = M(j, \lceil \log_2(T) \rceil+1)$. This is because, the set $\mathcal{T}_{j}$ contains all the leaf nodes of the recursive tree generated by the execution of Algorithm~\ref{alg:find_bp}. (See Line~\ref{line:firstformain} in Algorithm~\ref{alg:main}).   
	In addition, Algorithm~\ref{alg:restrict}, can only decrease the number of intervals belonging to $\mathcal{I}_{j+1}$ compared to the number of tuples belonging to $\mathcal{T}_{j}$. 
	Thus, thanks to the latter observations, we have that:
	\begin{equation*}
		|\mathcal{I}_{j+1}| \le |\mathcal{T}_{j}|  \le n \lceil \log_2(T)\rceil +   |\mathcal{I}_{j}|
	\end{equation*}
	and $|\mathcal{I}_{1}|=1$.
	Then, by induction, we can easily show that:
	\begin{equation*}
		|\mathcal{I}_{j}|  \le  (j-1) n \lceil \log_2(T)\rceil + 1 \le  j n \lceil \log_2(T)\rceil ,
	\end{equation*}
	for each $j \in \mathbb{N}$.
	We prove that, for each $j \ge 1$, the number of calls to Algorithm~\ref{alg:find_bp} satisfies:
	\begin{equation*}
		\sum_{k \in \mathbb{N}_{+}} N(j,k) = N( j, 1) + \sum_{k =2}^{\lceil \log_2(T)\rceil+1} N(j,k) \le |\mathcal{I}_{j}| + 2 n \lceil \log_2(T)\rceil \le  (j+2) n \lceil \log_2(T) \rceil,
	\end{equation*}
	concluding the proof.
\end{proof}

To prove Lemma~\ref{lem:call_find_bp}, we observe that in each epoch $j \in \mathbb{N}_{+}$, thanks to Lemma~\ref{lem:find_bp_terminate}, \texttt{Find-Jumps} recursively invokes itself only if the distributions at the extremes of the input interval are different and, thus, if the input interval contains at least one jump discontinuity. 
Since the number of discontinuities is bounded by $n-1$, at each level $k \in \mathbb{N}$ of the tree of recursive calls generated by the execution of \texttt{Find-Jumps} in epoch $j \in \mathbb{N}_{+}$ (see Figure~\ref{fig:tree}), the procedure recursively invokes itself invoked at most $2(n-1)$ times.
%
%
Furthermore, given the termination condition at Lines~\ref{line:small_int_end} in Algorithm~\ref{alg:find_bp}, if the size of the input interval is smaller than $\nicefrac{1}{T}$, then the procedure does \emph{not} invoke itself anymore. 
Thus, by observing that \texttt{Find-Jumps} recursively invokes itself with an input interval whose length is halved, we have that the maximum depth $k \in \mathbb{N}_{+}$ reached in the tree of recursive calls at each epoch is at most $\log_2(T)$.
Therefore, by combining the above observations, the number of times that \texttt{Find-Jumps} is invoked during epoch $j \in \mathbb{N}_{+}$ is bounded by $2n\log_2(T) + |\mathcal{I}_{j}|$, as the procedure is initially invoked on each interval in $\mathcal{I}_{j}$ by Algorithm~\ref{alg:main}.
Finally, by employing an inductive argument, it is possible to show that $|\mathcal{I}_{j}| \le j n \log_2(T) $, thus providing an upper bound on the number of calls to \texttt{Find-Jumps} in each epoch $j \in \mathbb{N}_{+}$.

%
%
%

%
%
%
%
\begin{figure}[!htp]
	\centering
	 \scalebox{0.7}[0.7]{\input{tikz_plot/tree.tikz}}
	\caption{Example of tree of recursive calls generated by the execution of \texttt{Find-Jumps} (Algorithm~\ref{alg:find_bp}).}
	\label{fig:tree}
\end{figure}
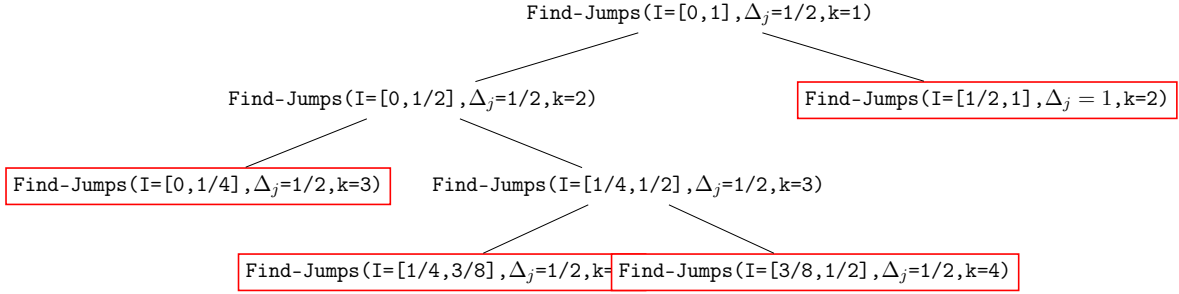

We are now ready to prove the main theorem of this section (and of the paper).

\begin{restatable}{theorem}{MainAlgoSecond}\label{thm:MainThm}
	The regret suffered by Algorithm~\ref{alg:main} is bounded by ${\mathcal{O}}(\sqrt{nT}\log^2(T))$.
\end{restatable}
\begin{proof}
	Let $\eta >0$. We observe that at iteration $j \in \mathbb{N}_{+}$, with probability $1 - \eta$, by applying Hoeffding's inequality and a union bound, Algorithm~\ref{alg:find_bp} computes two estimates satisfying:
	\begin{equation}\label{eq:hoeffding}
		|\widehat{\mu}(\alpha_i) - \mu_{h(\alpha_i)} | \le \frac{\Delta_j}{4},
	\end{equation}
	for each $i = 1, 2$, using $\nicefrac{8}{\Delta_j^2} \log( \nicefrac{4}{\eta} )$ samples for each extreme of the interval.
	
    Thus, employing a union bound and observing that Algorithm~\ref{alg:find_bp} is called at most $T$ times, we have that the condition in Equation~\ref{eq:hoeffding} is satisfied with probability $1 - \eta T$ during the execution of Algorithm~\ref{alg:main}.
	By setting $\eta T = \delta$, we ensure that with $\nicefrac{8}{\Delta_j^2} \log( \nicefrac{4T}{\delta})$ samples from both extremes, the condition in Equation~\ref{eq:hoeffding} is satisfied with probability $1 - \delta$ whenever Algorithm~\ref{alg:find_bp} is invoked.
	Thus, $\mathbb{P}(\mathcal{E})\ge 1-\delta$.

	In the following, we denote with $\mathcal{P}_j$ the number of rounds required to execute iteration $j \in \mathbb{N}_{+}$.
	We notice that $\sum_{j \in \mathbb{N}_{+}}\mathcal{P}_j = T$, and, by Lemma~\ref{lem:call_find_bp}, we have:
	\begin{equation*}
		\mathcal{P}_j \le \underbrace{  2 \frac{8}{\Delta^2_j} \log\left(\nicefrac{4T}{\delta}\right)}_{(a)}\,\, \underbrace{ (j+2) n \lceil \log_2(T) \rceil}_{(b)} =   \frac{16}{\Delta^2_j} \log\left(\nicefrac{4T}{\delta}\right) (j+2) n \lceil \log_2(T) \rceil
	\end{equation*}
	Where (a) represents an upper bound to number of samples required to execute Algorithm~\ref{alg:find_bp} at iteration  $j \in \mathbb{N}_{+}$, while (b) is a upper bound to the number of calls to Algorithm~\ref{alg:find_bp} at iteration $j \in \mathbb{N}_{+}$ according to Lemma~\ref{lem:call_find_bp}.

	Furthermore, under the clean event $\mathcal{E}$, at each iteration $j \in \mathbb{N}_{+}$, the learner selects actions that are $(4\Delta_{j-1} + \nicefrac{2}{T})$-optimal, as shown in Lemma~\ref{lem:deltaopt}.
	Since $\Delta_{j-1} = 2\Delta_{j}$, and $\Delta_{j} = \nicefrac{1}{2^j}$, this implies that the actions selected by the learner are $(8\Delta_{j} + \nicefrac{2}{T})$-optimal.
	Consequently, the cumulative regret can be upper bounded as follows:
	\begin{align}
		R_T & \le \left(\sum_{j \in \mathbb{N}_{+}}\mathcal{P}_j \left(8\Delta_{j} + \frac{2}{T}\right) \right) \mathbb{P}(\mathcal{E}) + T \cdot \mathbb{P}(\mathcal{E}^\mathsf{C})  \nonumber \\ 
		&\le  \sum_{j \in \mathbb{N}_{+}}8\Delta_{j} \mathcal{P}_j + 2 + T \cdot \delta   \nonumber \\
		& =  \underbrace{\sum_{j \, \le \, j^\star}8\Delta_{j}\mathcal{P}_j }_{(*)}  +  \underbrace{\sum_{j > j^\star}8\Delta_{j}\mathcal{P}_j }_{(**)}+ 2 + T \cdot \delta, \label{eq:regret_bound_first}
	\end{align}
	where we let $j^\star = \max \left\{0, \left\lfloor \log_2 \left(\sqrt{\frac{T}{n}}\frac{1}{\log_2^2(T)}\right)\right\rfloor \right\}$.

	We separately bound the two terms $(*)$ and $(**)$ in Equation~\ref{eq:regret_bound_first}.
	\begin{itemize}
		\item Bound on $(*)$. Thanks to the definition of $j^\star \in \mathbb{N}_+$, we have $\Delta_j \ge \log_2^2(T) \sqrt{\nicefrac{n}{T}}$ for each $j \in [\, j^\star]$.
		Thus, we have:
		\begin{align*}
			\sum_{j \, \le \, j^\star}8\Delta_{j}  \mathcal{P}_j
			&\le \sum_{j \, \le \, j^\star}  \frac{128}{\Delta_j} \log\left(\nicefrac{4T}{\delta}\right) (j+2) n \lceil \log_2(T) \rceil \\
			&\le \frac{{128} \sqrt{nT} \log\left(\nicefrac{4T}{\delta}\right)  \lceil \log_2(T) \rceil \sum_{j \, \le \, j^\star}(j+2) }{  \log^2_2(T)}  \\
			&\le \frac{{128} \sqrt{nT} \log\left(\nicefrac{4T}{\delta}\right)  \lceil \log_2(T) \rceil \left(  \log^2_2(T) + 3 \log_2(T) \right) }{  \log^2_2(T)}  \\
			&\le \mathcal{O} \left( \sqrt{nT} \log\left(\nicefrac{4T}{\delta} \right)   \log(T) \right).
		\end{align*}
		\item Bound on $(**)$. Thanks to the definition of $j^\star \in \mathbb{N}_+$, it holds $\Delta_j \le \log_2^2(T) \sqrt{\nicefrac{n}{T}}$ for each $j >  j^\star$.
		Thus, we can prove that the following holds:
		\begin{align*}
			\sum_{j > j^\star}8\Delta_{j}  \mathcal{P}_j \le 8 \sqrt{\frac{n}{T}} \log_2^2(T ) \sum_{j > j^\star}   \mathcal{P}_j \le 8 \sqrt{nT} \log^2_2(T) \le \mathcal{O}(\sqrt{nT} \log^2(T)).
		\end{align*}
	\end{itemize}
	Putting $(*)$ and $(**)$ together and setting $\delta=\nicefrac{1}{T}$ we have:
	\begin{align*}
		R_T \le  \mathcal{O}(\sqrt{nT} \log^2(T)),
	\end{align*}
	concluding the proof.
\end{proof}

We provide a sketch of how to prove that the regret suffered by Algorithm \ref{alg:main} is of order $\widetilde{\mathcal{O}}(\sqrt{T})$, where we explicitly omit the dependence on $n$ for ease of presentation.
Notice that, thanks to Lemma~\ref{lem:deltaopt}, the regret suffered during the execution of \texttt{Find-Jumps} in epoch $j \in \mathbb{N}_{+}$ is $\mathcal{O}(\nicefrac{1}{\Delta_{j}})=\mathcal{O}({2^{j}})$.
Furthermore, thanks to Lemma~\ref{lem:call_find_bp}, we have that \texttt{Find-Jumps} is invoked a number of times proportional to the index of the epoch $j \in \mathbb{N}_{+}$.
Therefore, the regret suffered during the execution of \texttt{Find-Jumps} in epoch $j \in \mathbb{N}_{+}$ is bounded by $\mathcal{O}(j2^j)$.
Let us also observe that the largest epoch $j^\star \in \mathbb{N}_{+}$ executed by Algorithm~\ref{alg:main} must be such that $j^\star \le {\log_2(T)}/{2}$.
Otherwise, the \texttt{Find-Jumps} procedure would require a number of rounds larger than $\nicefrac{1}{\Delta^2_{j^\star}}= T$ to be executed.
Therefore, the regret suffered by \texttt{Find-Jumps} can be bounded by:
$$
	\sum_{j\le {\log_2(T)}/{2}} j 2^j \le \widetilde{\mathcal{O}}(\sqrt{T}).
$$
In conclusion, thanks to Theorem~\ref{thm:MainThm} and the relations between BwMJ problems and microeconomic models discussed in Section~\ref{sec:applications}, we have that the two following corollaries hold.
\begin{corollary}\label{cor:contracts}
In (non-Bayesian) hidden-action principal-agent problems, there exists an algorithm that achieves a regret of $\widetilde{\mathcal{O}}(\sqrt{nT} )$ when competing against an optimal linear contract, where $n$ is the number of agent's actions.
Moreover, in Bayesian hidden-action principal-agent problems, there exists an algorithm that achieves a regret of $\widetilde{\mathcal{O}}(\sqrt{ndT} )$ when competing against an optimal linear contract, where $d \in \mathbb{N}$ is the number of possible agent's types.
\end{corollary}
In hidden-action principal-agent problems, \citet{zhu2023sample} proposed an algorithm that achieves a regret of $\widetilde{\mathcal{O}}(T^{\nicefrac{2}{3}})$ when competing against an optimal linear contract.
\citet{zhu2023sample} leave open the question of whether such regret dependence could be improved when the number of the agent's actions is sufficiently smaller compared to the time horizon $T$. %
We provide a positive answer to this question in Corollary~\ref{cor:contracts}.
%
%
%
%
%
%
%
\begin{corollary}\label{cor:auctions}
	In posted-price auctions with finitely many buyer valuations, there exists an algorithm that achieves a regret of $\widetilde{\mathcal{O}}(\sqrt{nT} )$ when competing against an optimal fixed price, where $n$ is the number of possible buyer valuation values.
\end{corollary}
In posted-price auctions with finitely many buyer valuations, \citet{cesa2019dynamic} proposed an algorithm that achieves a regret of $\widetilde{\mathcal{O}}(\sqrt{nT} + V(V+1))$, where $V$ is an instance-dependent parameter encoding the magnitude of jump gaps.
\citet{cesa2019dynamic} leave open the question of whether the dependence on such a parameter $V$ is avoidable or not. 
We positively answer this question in Corollary~\ref{cor:auctions}.

%% file: tikz_plot/tree.tikz
\begin{tikzpicture}[
	every node/.style={align=center},
	level 1/.style={sibling distance=10cm},
	level 2/.style={sibling distance=7.5cm},
	level 3/.style={sibling distance=6.5cm},
	leaf/.style={rectangle, draw=red, fill=white, thick}
	]
	\node (A) at (0,0) {\texttt{Find-Jumps(I=[0,1],$\Delta_j$=1/2,k=1)}}
	child {node (B) {\texttt{Find-Jumps(I=[0,1/2],$\Delta_j$=1/2,k=2)}}
		child {node[leaf] (D) {\texttt{Find-Jumps(I=[0,1/4],$\Delta_j$=1/2,k=3)}}}
		child {node (E) {\texttt{Find-Jumps(I=[1/4,1/2],$\Delta_j$=1/2,k=3)}}
			child {node[leaf] (G) {\texttt{Find-Jumps(I=[1/4,3/8],$\Delta_j$=1/2,k=4)}}}
			child {node[leaf] (F) {\texttt{Find-Jumps(I=[3/8,1/2],$\Delta_j$=1/2,k=4)}}}
		}
	}
	child {node[leaf] (C) {\texttt{Find-Jumps(I=[1/2,1],$\Delta_j=1$,k=2)}}};
\end{tikzpicture}

%% file: lower_bound.tex
\section{On the Optimality of the \texttt{RJI-OS} Algorithm}\label{sec:lowerbound}

In this section, we show that the regret guarantees attained by the \texttt{RJI-OS}  algorithm are optimal.
Specifically, the regret attained by Algorithm~\ref{alg:main} is optimal even when restricting the attention to BwMJ instances related to microeconomic settings, namely those capturing hidden-action principal-agent problems and posted-price auctions, which have been introduced in Section~\ref{sec:applications}.

Specifically, the \texttt{RJI-OS} algorithm, when employed in posted-price auctions settings (see Section~\ref{sec:applications_postedprice}), achieves optimal regret guarantees that match the lower bound introduced by~\citet{cesa2019dynamic}.
Similarly, when employed in principal-agent problems (see Section~\ref{sec:applications_contracts}), the \texttt{RJI-OS} algorithm attains optimal regret guarantees, matching a lower bound we discuss in the following.

We now present a lower bound for principal-agent settings.
Specifically, for the problem of minimizing the regret with respect to an optimal-in-hindsight \emph{linear contract} (see Section~\ref{sec:applications_contracts} for its formal introduction).
%
%
Our result is formally stated in the following Theorem~\ref{thm:lb}.
Notice that Theorem~\ref{thm:lb} strengthens an existing lower bound by~\citet{zhu2023sample}, which holds only when the agent has a number of actions proportional to $T^{\nicefrac{1}{3}}$, whereas Theorem~\ref{thm:lb} holds for any number of agent's actions smaller than or equal to that needed by~\citet{zhu2023sample}.
Formally:
\begin{restatable}{theorem}{LowerBound}\label{thm:lb}
	Let $n, T \in \mathbb{N}_{+}$ with $T^{\nicefrac{1}{3}} \ge n \ge 3$.
	Then, for any learner's policy $\pi \coloneqq (\pi_t)_{t \in [T]}$, there exists an instance of BwMJ problem, encoding a hidden-action principal-agent problem where the agent has $n$ actions available (see Section~\ref{sec:applications_contracts}), in which the regret attained by the policy $\pi$ is
	\[
		R_T (\pi) \ge \Omega \left( {\sqrt{nT}} \right).
	\]
	%
	%
	%
\end{restatable}
\begin{proof}
	Let $\epsilon \in (0,1)$ and $k \in ( 2, 3) $ be two parameters to be defined in the following. 
	We consider a principal-agent instance $\mathcal{P}$ characterized by two outcomes $ r_1 = 1 $ and $ r_2 = 0 $ and by $ n \ge 3 $ agent's actions.
	The costs and the corresponding probability distributions over outcomes of these actions are defined as follows.
	\begin{align*}
		\begin{cases}
			F_{1}=\left(0,1 \right) \,\, & c_1 = 0 \\
			\vspace{1.5mm}
			F_{2}=\Big(\frac{1}{2} + \frac{ 1 }{2} \epsilon ,  \frac{1}{2} - \frac{1 }{2} \epsilon \Big) \,\, & c_2 = \left( 1+ \epsilon \right)\left(\frac{1}{2} - \frac{1}{k}  \right ) \\
			F_{i}=\Big(   \frac{1}{2} + \epsilon (i-2) , \frac{1}{2} - \epsilon (i-2)  \Big) \,\ & c_{i} =  c_{i-1} + \alpha_i \left(F_{i,1} - F_{i-1,1} \right) \,\,\,\,\, i \ge 3,
		\end{cases}
	\end{align*}
	Where the parameters $\alpha_i \in [0,1]$ are given by:
	\begin{equation*}
		\alpha_1 = 0 \quad \textnormal{and}	\quad \alpha_i \coloneqq 1 - \frac{1}{\left( \frac{1}{2} + \epsilon (i-2) \right)k }  \,\,\,\,\, i=2, \dots, n.
	\end{equation*}
	In the following, we assume $\epsilon n \le \nicefrac{1}{4}$ to ensure that the probability distributions and the costs of the different actions are well-defined. Indeed, for each $ i \ge 3$, we have:
	\begin{equation*}
		c_{i} =  c_{i-1} + \alpha_i \left(F_{i,1} - F_{i-1,1} \right) \le c_{i-1} + \epsilon.
	\end{equation*}
	The last inequality holds observing that $|\alpha_i|\le 1$ for each $i \in [n]$ and $|F_{i,1} - F_{i-1,1} | \le \epsilon$ for each $i \ge 3$.
	Thus, thanks to the assumption that $\epsilon n \le \nicefrac{1}{4}$, it is easy to verify that $c_n \le 1$ by employing a simple recursive argument and noticing that $c_2 \le \nicefrac{1}{3}$ since $\epsilon \le 1$ and $k \le 3$.

	We also notice that each action $i \ge 2$ provides the agent with an expected utility greater or equal to what they would achieve by selecting action $i-1$ when the principal commits to a contract $\alpha \ge \alpha_i$.
	This is because $F_{i,1} \ge F_{i-1,1}$ for each $i \ge 2$, and, by a simple calculation, it is possible to verify that the following equalities hold:
	\begin{equation*}
		F_{i-1,1} \alpha_{i}  - c_{i-1} = F_{i,1} \alpha_{i}  - c_{i} \,\,\ \textnormal{for each} \,\,\, i=2, \dots, n.
	\end{equation*}
	%
	%
	Thus, each action $i \ge 2$ provides the agent with a higher expected utility compared to all the actions $i' \in [n]$ with $i' < i$ when $\alpha \ge \alpha_i$.
	At the same time, the action $i+1$ provides the agent higher utility compared to action $i$ only when $\alpha \ge \alpha_{i+1}$.
	Consequently, the agent's best-response regions are given by:
	\begin{equation}\label{eq:br_regions_lb}
		\mathcal{B}_i=\left[\alpha_{i}, \, \alpha_{i+1} \right)  \,\ i=1, \dots , n-1 \,\,\,\,\,\,\ \textnormal{and} \,\,\,\,\,\,\,\	\mathcal{B}_n=\left[\alpha_{n}, \, 1 \right],
	\end{equation}
	as we assume the agent breaks ties in favor of the principal.

	We also observe that for each $\alpha \in \mathcal{B}_i$ and action $i \in [n]$, the principal's expected utility satisfies $u(\alpha) \le u(\alpha_i)$, as $u(\alpha)$ is a decreasing function in $\mathcal{B}_i$.
	Furthermore, the principal's expected utility in each $\alpha_i$, with $i \ge 3$, is such that:
	\begin{equation*}
		u(\alpha_i) = (1- \alpha_i) F_{i, 1} = \frac{1}{\left(\frac{1}{2} + \epsilon (i-2) \right) k } \left(\frac{1}{2} + \epsilon (i-2) \right) = \frac{1}{k}.
	\end{equation*}
	Similarly, when $\alpha=\alpha_2$, the principal's expected utility is equal to:
	\begin{equation*}
		u(\alpha_2) = (1- \alpha_2) F_{2, 1} = \frac{2}{k} \left(\frac{1+ \epsilon}{2} \right) = \frac{1 + \epsilon}{k}.
	\end{equation*}
	Finally, we notice that $u(\alpha_1)=0$.
	Thus, in instance $\mathcal{P}$, the optimal (linear) contract is equal to $\alpha_2$. 
	
	%
	We denote with $\mathbb{P}$ the probability distribution over all possible histories associated with the execution of $\pi$ in instance $\mathcal{P}$ over $T \ge 0$ rounds.
	Similarly, we define $\mathbb{E}$ as the expectation under the probability distribution $\mathbb{P}$.
	We also denote with $T_i(T)$ the number of rounds in which the agent selects action $i \in [n]$ out of $T \ge 0$ rounds, and we let:
	\begin{equation}\label{eq:istar_lb}
		i^\star = \arg\min_{i>2} \mathbb{E}\left[ \,T_i(T) \, \right].
	\end{equation}
	Intuitively, the action $i^\star > 2$ corresponds to the (sub-optimal) action that is selected less frequently by the agent when the policy $\pi$ is executed in instance $\mathcal{P}$ over $T \geq 0$ rounds.
	
	We now introduce a different principal-agent instance $ \mathcal{P}'$ depending on the action $i^\star$ defined in Equation~\ref{eq:istar_lb}.
	This new instance shares the same distribution over the set of outcomes as instance $\mathcal{P}$, with the main difference being that, in instance $\mathcal{P}'$, the probability of observing the first outcome when the agent selects action $i^\star$ is increased by $\epsilon > 0$.
	Formally, we have:
	\begin{align*}
		\begin{cases}
			F_{1}'=\left(0,1 \right) \,\, & c_1' = 0 \\
			\vspace{1.5mm}
			F_{2}'=\Big(\frac{1}{2} + \frac{ 1 }{2} \epsilon ,  \frac{1}{2} - \frac{1 }{2} \epsilon \Big) \,\, & c_2' = \left( 1+ \epsilon \right)\left(\frac{1}{2} - \frac{1}{k}  \right ) \\
			\vspace{1.5mm}
			F_{i^\star}'=\Big(   \frac{1}{2} + \epsilon (i^\star-1), \frac{1}{2} - \epsilon (i^\star-1)  \Big) \,\ & c_{i^\star}' =  c_{i^\star-1}' + \alpha_{i^\star} \left(F_{i^\star,1}' - F_{i^\star-1,1}' \right)\\
			F_{i}'=\Big(   \frac{1}{2} + \epsilon (i-2) , \frac{1}{2} - \epsilon (i-2)  \Big) \,\ & c_{i}' =  c_{i-1}' + \alpha_i \left(F_{i,1}' - F_{i-1,1}' \right) \,\,\,\,\,\,\, i\ge 3. \\
		\end{cases}
	\end{align*}
	We observe that, in instance $\mathcal{P}'$, the actions $i^\star$ and $i^\star + 1$ coincide when $i^\star \neq n$.
	Therefore, it is without loss of generality to assume that when the principal selects a contract $\alpha \in \mathcal{B}_{i^\star}$, the agent chooses action $i^\star$, while when the principal selects $\alpha \in \mathcal{B}_{i^\star+1}$, the agent chooses action $i^\star + 1$.
	Using the same argument employed above, it is possible to verify that the agent's best-response regions $\mathcal{B}_{i}'$ in instance $\mathcal{P}'$ satisfy:
	\begin{equation*}
		\mathcal{B}_{i}'=\mathcal{B}_{i} \quad \forall i\in [n].
	\end{equation*}
	%
	%
	%
	%
	Moreover, we notice that the principal's utility in instance $\mathcal{P}'$, evaluated at $\alpha_{i^\star}$, is equal to:
	\begin{align*}
		u'(\alpha_{i^\star}) = (1- \alpha_{i^\star}) F_{{i^\star}, 1} & = \frac{1}{\left(\frac{1}{2} + \epsilon ( {i^\star}-2) \right) k } \left(\frac{1}{2} + \epsilon ( {i^\star}-2) + \epsilon \right)\\ &= \frac{1}{k} \left( 1 + \frac{\epsilon}{\left(\frac{1}{2} + \epsilon ( {i^\star}-2) \right)}\right)\\
		& \ge \frac{1}{k} \left( 1 + \frac{4}{3}\epsilon\right),
	\end{align*}
	where the last inequality holds since $\epsilon \le \epsilon n \le \nicefrac{1}{4}$. 
	Furthermore, the principal's utility evaluated at the remaining $\alpha_{i}$ is equal to:
	\begin{equation*}
		u'(\alpha_{i}) = u(\alpha_{i}),
	\end{equation*}
	for each $i \neq i^\star$. Thus, in instance $\mathcal{P}'$, the optimal contract coincides with $\alpha_{i^\star}$.  
	In the following, we assume that whenever the principal selects a contract $\alpha \in \mathcal{B}_i$ (or $\alpha \in \mathcal{B}_i'$), they achieve an expected utility equal to what they achieve by selecting $\alpha_i$.
	This assumption only facilitates the learning problem faced by the principal, since the principal's expected utility is a decreasing function in each agent's best-response region.
	Moreover, the feedback received by the principal is equivalent for each action $\alpha \in \mathcal{B}_i$ (or $\alpha \in \mathcal{B}_i'$).
	%
	%
	%
	%
	
	
	In this way, the learning problem faced by the principal reduces to a multi-armed bandit problem, where the set of arms available in the two instances coincides with the set $\{\alpha_i\}_{i=1}^n$, and each arm $\alpha_i$ provides utility $u(\alpha_i)$ and $u'(\alpha_i)$ in instances $\mathcal{P}$ and $\mathcal{P}'$, respectively.
	%
	Thus, the regret suffered in instance $\mathcal{P}$ can be lower bounded as follows:
	\begin{equation*}
		R_T(\pi) \ge  \mathbb{P} \left(T_2(T)\le T/2\right) \frac{T}{2} \frac{\epsilon}{k},
	\end{equation*}
	since the optimal arm in instance $\mathcal{P}$ provides an expected utility to the principal that is ${\epsilon}/{k}$ better compared to all the contracts $\alpha_i $ with $i \neq 2$.
	Furthermore, the regret in instance $\mathcal{P}'$ can be lower-bounded as follows.
	\begin{equation}\label{eq:regret_lb_2}
		R'_T(\pi) \ge \mathbb{P}' \left(T_2(T) > T/2\right) \frac{T}{2} \frac{\epsilon}{3 k},
	\end{equation}
	where, similar to instance $\mathcal{P}$, we let $\mathbb{P}'$ be the probability distribution over all possible histories when executing $\pi$ in instance $\mathcal{P}'$ over $T \ge 0$ rounds and we define $\mathbb{E}'$ the expectation under~$\mathbb{P}'$.
	
	We observe that Equation~\eqref{eq:regret_lb_2} holds because $\alpha_2$ provides an expected utility that is ${\epsilon}/{3k}$ worse compared to the utility achieved in $\alpha_{i^\star}$. 
	We also observe that the average number of times the agent selects action $i^\star$ when the principal executes policy $\pi$ in instance $\mathcal{P}$ can be upper bounded as follows: 
	\begin{equation*}
		\mathbb{E} [T_{i^\star}(T)] \le \frac{T}{n-2}.
	\end{equation*}
	Since $\sum_{i \ge 3 } \mathbb{E} [T_{i}(T)] \le T$ and because of the definition of $i^\star$ introduced in Equation~\eqref{eq:istar_lb}.
	Furthermore, by employing the Bretagnolle–Huber inequality, we have that:
	\begin{align*}
		R_T (\pi)+ R'_T(\pi) & \ge
		\left( \mathbb{P} \left(T_2(T)\le T/2\right) +  \mathbb{P}' \left(T_2(T) > T/2\right) \right) \frac{T \epsilon}{6 k}\\
		& \ge \frac{T \epsilon}{12 k} e^{-\mathcal{KL}(\mathbb{P}, \mathbb{P}')}.
	\end{align*}

	Then, by employing the divergence decomposition lemma (see~\cite{lattimore2020bandit}) and observing that the principal receives different feedback in the two instances only when they select contracts $\alpha \in \mathcal{B}_{i^\star}$, we have:
	\begin{align*}
		\mathcal{KL} \left(\mathbb{P}, \mathbb{P}' \right) & = \mathbb{E} [T_{i^\star}(T)] \cdot\mathcal{KL} \left(F_{i^\star}, F_{i^\star} ' \right)\\
		& \leq \frac{T }{ n-2}  \frac{\epsilon^2}{\left(\frac{1}{2}- \epsilon (i^\star+1) \right) \left(\frac{1}{2} + \epsilon (i^\star-1) \right)} \\
		& \leq \frac{16}{3} \frac{T }{ n-2}  \epsilon^2,
	\end{align*}
	where we employed the fact that the Kullback-Leibler divergence between two Bernoulli distributions can be upper bounded as $\mathcal{KL}(\mathcal{B}(p), \mathcal{B}(q)) \le (p-q)^2/((1-q)q)$ and the fact that $\epsilon( i^\star-1) \le \epsilon n \le \nicefrac{1}{4}$. Thus, we have:
	\begin{align*}
		\max \{ R_T (\pi),R'_T(\pi)\} &\ge \frac{1}{2} \left( R_T (\pi)+ R'_T(\pi) \right) \\
		& \ge \frac{T \epsilon}{24 k} e^{-\mathcal{KL}(\mathbb{P}, \mathbb{P}')}\\
		& \ge \frac{T \epsilon}{24 k} \exp \left( -  \frac{16 T \epsilon^2}{3 (n-2)} 	\right)\\
		& \ge \frac{5 e^{-1}}{1008} \sqrt{nT}
	\end{align*}
	Where the last step holds by setting $\epsilon = \sqrt{\nicefrac{n}{16 T}} $, $k =\nicefrac{21}{10} $ and observing that $n/(n-2) \le 3$ since $n \ge 3$. Finally, we have that the number of actions $n$ must satisfy:
	\begin{equation*}
		\epsilon n =  \sqrt{\frac{n}{16 T}} n  \le \frac{1}{4},
	\end{equation*}
	which implies that $n \le T^{\nicefrac{1}{3}} $, concluding the proof.
\end{proof}
%
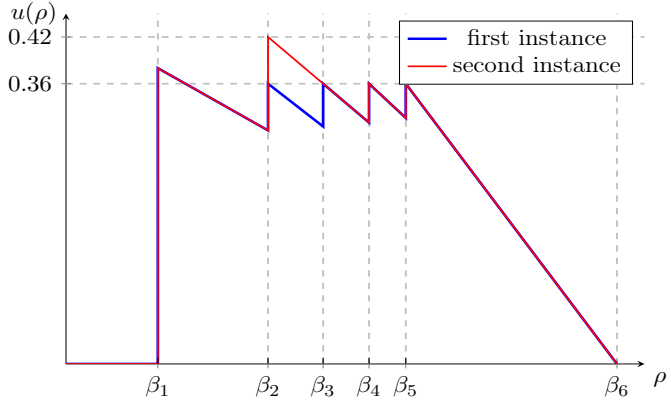
\begin{figure}[!htp]
	\centering
	\resizebox{0.6\linewidth}{0.37\linewidth}{\input{lower_bound.tikz}}
	\caption{Principal's expected utility in the two instances used to prove the lower bound in Theorem~\ref{thm:lb}.}
	\label{fig:lower_bound}
\end{figure}
In order to prove Theorem~\ref{thm:lb}, we draw a suitable connection between regret-minimizing linear contracts in principal-agent problems and regret minimization in classical multi-armed bandits.
%
%
To do so, we construct two principal-agent instances in which the distributions over outcomes of agent's actions coincide, except for one.
Thus, the principal receives the same feedback and utility, except when the agent selects the action that is different in the two instances.
This is made possible by suitably choosing the costs of agent's actions in the two cases.
Figure~\ref{fig:lower_bound} provides an example of the principal’s expected utility in the two instances.
We notice that, in order to distinguish between the two instances, the principal is required to play in the interval $\mathcal{B}_{2}=[\beta_2, \beta_3)$, since this corresponds to the set of linear contracts under which the action different in the two instances is a best response for the agent.
However, this results in a large regret in the first instance, since an optimal contract in the first instance is $\beta_1\in[0,1]$.
Similarly, if the principal only plays a few samples from the interval $\mathcal{B}_{2}=[\beta_2, \beta_3)$, then the regret in the second instance is large, since an optimal contract in such an instance is $\beta_2\in[0,1]$.

Let us observe that the lower bound presented in Theorem~\ref{thm:lb} cannot be extended to settings with a larger number of actions, except for absolute constants and logarithmic factors in the number of rounds $T$.
Indeed, if in Theorem~\ref{thm:lb} the number of an agent's actions $n$ exceeds $T^{\nicefrac{1}{3}}$, then the lower bound on the regret would be $\Omega(T^{\alpha})$ for some $\alpha > \nicefrac{2}{3}$.
%
%
However, this contradicts the $\widetilde{{\mathcal{O}}}(T^{\nicefrac{2}{3}})$ upper bound on the regret achievable by employing the algorithm by~\citet{kleinberg2004nearly}.
%

%


We conclude this section by discussing the results we obtained in light of the lower bound presented in this section.
In Figure~\ref{fig:regret_comparison} (\emph{Left}), we show that in principal-agent problems, the regret of Algorithm~\ref{alg:main} matches the lower bound presented in Theorem~\ref{thm:lb} (up to logarithmic factors) when the number of agent's actions satisfies $n \leq T^{\nicefrac{1}{3}}$, while the regret bound achieved by the algorithm in~\cite{zhu2023sample} could be largely suboptimal.
Additionally, as shown in Figure~\ref{fig:regret_comparison} (\emph{Right}), the regret suffered by Algorithm~\ref{alg:main} matches (up to logarithmic factors) the $\Omega(\sqrt{nT})$ lower bound by \citet{cesa2019dynamic} for posted-price auction settings.
This is not the case for the regret suffered by the algorithm proposed by \citet{cesa2019dynamic}, which can be significantly larger in scenarios in which identifying the jump discontinuities is challenging.
\begin{figure}[!t]
	\centering
	\resizebox{0.85\linewidth}{!}{\input{tikz_plot/regret_comparison.tikz}}
	\caption{Regret upper and lower bounds as functions of the parameter \(n \le T^{\nicefrac{1}{3}}\). For the sake of presentation, we omitted logarithmic factors in the dependence of the regret bounds suffered by the different algorithms.}
	\label{fig:regret_comparison}
\end{figure}
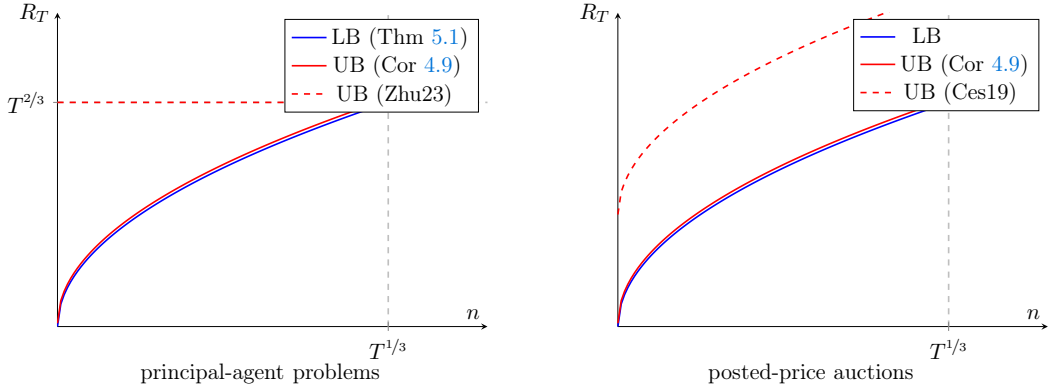

%% file: lower_bound.tikz
	\begin{tikzpicture}
	\begin{axis}[
		axis lines=middle,
		xlabel={$\rho$},
		xlabel style={below right},
		ylabel={$u(\rho)$},
		ylabel style={left},
		xmin=0, xmax=0.63,
		ymin=0, ymax=0.45,
		xtick={0,1},
		ytick={0},
		grid=both,
		grid style={dashed,gray!50},
		extra y ticks={ 0.36,0.42,0.5,0.6},
		extra y tick style={grid=major, grid style={thick, dashed}},
		extra x ticks={0.1,0.22, 0.28, 0.33, 0.37, 0.6},
		extra x tick labels={$\beta_1$, $\beta_2$, $\beta_3$, $\beta_4$, $\beta_5$, $\beta_6$, $\beta_7$},
		extra x tick style={grid=major, grid style={thick, dashed}},
		width=10.5cm,
		height=7.5cm
		]
		\addplot[ blue, very thick] coordinates {
			(0.0, 0.0)
			(0.1, 0.0)
			(0.1, 0.38)
			(0.22, 0.3)
			(0.22, 0.36)
			(0.28, 0.305)
			(0.28, 0.36)
			(0.33, 0.31)
			(0.33, 0.36)
			(0.37, 0.316)
			(0.37, 0.36)
			(0.6, 0.0)
		};
		\addlegendentry{\textnormal{first instance}}
		
		\addplot[red, thick] coordinates {
			(0.0, 0.0)
			(0.1, 0.0)
			(0.1, 0.38)
			(0.22, 0.3)
			(0.22, 0.42)
			(0.22, 0.42)
			(0.33, 0.31)
			(0.33, 0.36)
			(0.37, 0.316)
			(0.37, 0.36)
			(0.6, 0.0)
		};
		\addlegendentry{\textnormal{second instance}}
		
		\draw[dashed, line width=0.02mm] (axis cs:0.1,0) -- (axis cs:0.1,0.01);
		\draw[dashed, line width=0.02mm] (axis cs:0.22,0) -- (axis cs:0.22,0.01);
		\draw[dashed, line width=0.02mm] (axis cs:0.28,0) -- (axis cs:0.28,0.01);
		\draw[dashed, line width=0.02mm] (axis cs:0.33,0) -- (axis cs:0.33,0.01);
		\draw[dashed, line width=0.02mm] (axis cs:0.37,0) -- (axis cs:0.37,0.01);
	\end{axis}
\end{tikzpicture}

%% file: tikz_plot/regret_comparison.tikz
\begin{tikzpicture}
	\begin{axis}[
		axis lines=middle,
		xlabel={$n$},
		ylabel={$R_T$},
		ylabel style={left},
		xmin=0, xmax=1.3,
		ymin=0, ymax=1.4,
		xtick={0},
		ytick={0},
		grid=both,
		grid style={dashed,gray!50},
        extra y ticks={1},
		extra y tick style={grid=major, grid style={thick, dashed}},
        extra y tick labels={$T^{\nicefrac{2}{3}}$},
		extra x ticks={1},
		extra x tick labels={$T^{\nicefrac{1}{3}}$},
		extra x tick style={grid=major, grid style={thick, dashed}},
		width=9cm,
		height=7cm
		]
		
		\addplot[blue, thick, domain=0:1, samples=100] {sqrt(x)};
        \addlegendentry{\hspace{0mm}LB (Thm~\ref{thm:lb})}
        \addplot[red, thick, domain=0:1, samples=100] {sqrt(x)+0.015};
        \addlegendentry{\hspace{-1mm}\textnormal{UB (Cor~\ref{cor:contracts})}}
        \addplot[red, thick, dashed, domain=0:1, samples=100] {1}; \addlegendentry{\hspace{-1.9mm}\textnormal{UB (Zhu23)}}
	\end{axis}
    \node[below] at (3.5, -0.5) {{principal-agent problems}};

	\begin{axis}[
		name=plot2,
		at={(plot1.right of south east)},
		anchor=left of south west,
		xshift=3.0cm,
		axis lines=middle,
		xlabel={$n$},
		ylabel={$R_T$},
		ylabel style={left},
		xmin=0, xmax=1.3,
		ymin=0, ymax=1.4,
		xtick={0},
		ytick={0},
		grid=both,
		grid style={dashed,gray!50},
		extra y tick style={grid=major, grid style={thick, dashed}},
		extra x ticks={1},
		extra x tick labels={$T^{\nicefrac{1}{3}}$},
		extra x tick style={grid=major, grid style={thick, dashed}},
		width=9cm,
		height=7cm
		]
		
        \addplot[blue, thick, domain=0:1, samples=100] {sqrt(x)};
        \addlegendentry{\hspace{-13mm}\textnormal{LB }}
		\addplot[red, thick, domain=0:1, samples=100] {sqrt(x)+0.015};
        \addlegendentry{\textnormal{UB (Cor~\ref{cor:contracts})}}
        \addplot[red, thick, dashed, domain=0:1, samples=100] {sqrt(x)+0.5};
        \addlegendentry{\hspace{-1.5mm}\textnormal{UB (Ces19)}}

	\end{axis}
     \node[below] at (13.0, -0.5) {{posted-price auctions}};
	
\end{tikzpicture}